\documentclass[11pt,a4paper,reqno]{amsart}
\usepackage{fullpage}

\usepackage[utf8]{inputenc}

\usepackage{hyperref}
\usepackage{enumerate}
\usepackage{comment}

\usepackage{setspace}
\usepackage{amsfonts,amssymb,amsthm,amsmath,amsopn}
\usepackage{physics}
\usepackage{bbm}

\usepackage{tikz}
\usetikzlibrary{calc,decorations.pathreplacing}

\numberwithin{equation}{section}
\allowdisplaybreaks

\DeclareMathOperator{\supp}{supp}
\DeclareMathOperator{\dist}{dist}

\renewcommand{\i}{\mathrm{i}}
\newcommand{\e}{\mathrm{e}}
\renewcommand{\d}{\mathrm{d}}
\renewcommand{\op}{\mathrm{op}}

\newcommand{\bN}{\mathbb{N}}
\newcommand{\bR}{\mathbb{R}}
\newcommand{\bC}{\mathbb{C}}

\newcommand{\Id}{\mathbbm{1}}

\newcommand{\p}{\mathtt{p}}
\newcommand{\q}{\mathtt{q}}

\newcommand{\MF}{\mathrm{MF}}

\theoremstyle{plain} 
\newtheorem{theorem}{Theorem}[section]
\newtheorem*{theorem*}{Theorem}
\newtheorem{lemma}[theorem]{Lemma}
\newtheorem*{lemma*}{Lemma}
\newtheorem{corollary}[theorem]{Corollary}
\newtheorem*{corollary*}{Corollary}
\newtheorem{proposition}[theorem]{Proposition}
\newtheorem*{proposition*}{Proposition}

\theoremstyle{remark}

\newtheorem{remark}[theorem]{Remark}
\newtheorem*{remark*}{Remark}

\newtheorem*{remarks*}{Remarks}

\theoremstyle{definition} 

\newtheorem*{definition*}{Definition}

\newtheorem*{assumption*}{Assumption}
\newtheorem{hypothesis}[theorem]{Hypothesis}
\newtheorem*{hypothesis*}{Hypothesis}

\usepackage[dvipsnames]{xcolor}

\title[Propagation Estimates for the Many-body Boson Star Equation]{Propagation Estimates\\for the Many-body Boson Star Equation}

\author[Breteaux]{S\'ebastien Breteaux}
\address[S. Breteaux]{Universit\'e de Lorraine, CNRS, IECL, F-57000 Metz, France}
\email{\href{mailto:sebastien.breteaux@univ-lorraine.fr}{sebastien.breteaux@univ-lorraine.fr}}
\author[Dietze]{Charlotte Dietze}
\address[C. Dietze]{Sorbonne Universit\'e, CNRS, Universit\'e Paris Cit\'e, Laboratoire Jacques-Louis Lions (LJLL), 4 place Jussieu, F-75005 Paris, France}
\email{\href{mailto:charlotte.dietze@sorbonne-universite.fr}{charlotte.dietze@sorbonne-universite.fr}}
\author[Faupin]{J\'er\'emy Faupin}
\address[J. Faupin]{Universit\'e de Lorraine, CNRS, IECL, F-57000 Metz, France}
\email{\href{mailto:jeremy.faupin@univ-lorraine.fr}{jeremy.faupin@univ-lorraine.fr}}
\author[Grasselli]{\\Viviana Grasselli}
\address[V. Grasselli]{Universit\'e de Lorraine, CNRS, IECL, F-57000 Metz, France}
\email{\href{mailto:viviana.grasselli@univ-lorraine.fr}{viviana.grasselli@univ-lorraine.fr}}
\author[Lee]{Jinyeop Lee}
\address[J. Lee]{Department of Applied Mathematics, Kyung Hee University,
1732 Deogyeong-daero, Yongin-si, Gyeonggi-do, Republic of Korea}
\email{\href{mailto:jinyeop.lee@khu.ac.kr}{jinyeop.lee@khu.ac.kr}}

\date{\today}



\begin{document}

\begin{abstract}
We study the many-body boson star equation in the mean-field regime. Under a decay assumption on the solution to the effective one-body equation, we prove a uniform-in-time trace-norm estimate comparing the one-particle reduced density matrix of the many-body state with the projection onto the one-body solution. This establishes uniform-in-time mean-field convergence as the particle number goes to infinity. Furthermore, for the many-body dynamics, we establish both a maximal velocity bound, which can be seen as a Lieb--Robinson-type bound, and a minimal velocity bound. As an application, we recover maximal and minimal velocity bounds for the one-body boson star equation by passing to the mean-field limit.
\end{abstract}

\maketitle

\section{Introduction}\label{sec:intro}

\subsection{The Model}

We consider the many-body boson star equation in $\bR^3$
\begin{equation}\label{eq:NbdQSE}
    \i \partial_t \Psi_{N,t}
    =
    H_N\Psi_{N,t}, \quad \Psi_{N,t}\vert_{t=0}=\Psi_{N,0}
\end{equation}
where
\begin{equation}\label{eq:Nbody-Hamiltonian}
    H_N
    :=
    \sum_{j=1}^{N} \sqrt{-\Delta_j+m^2}
    +
    \frac{1}{N-1}
    \sum_{1\leq j<k\leq N} w(x_j-x_k)
\end{equation}
for some fixed $m>0$ and where $w:\bR^3\to\bR$ is an even function modeling interactions. For simplicity of exposition and without loss of generality (by a scaling argument), we set $$m=1$$ throughout the paper. Moreover, we  assume 
\begin{equation}\label{cond-w}
    w\in L^1(\bR^3)\cap L^\infty(\bR^3),
\end{equation}
though some of our results hold for more general $w$. 
In particular, \eqref{cond-w} implies that~$H_N$ identifies with a self-adjoint operator on $L^2(\mathbb{R}^{3N})$ with domain $H^1(\mathbb{R}^{3N})$ and form domain~$H^{1/2}(\mathbb{R}^{3N})$. We refer the reader to Subsection \ref{sec:hist} below for motivations to study Eq. \eqref{eq:NbdQSE} and references to previous works.

Everywhere in the paper, we consider \eqref{eq:NbdQSE} for bosonic initial states
\[
    \Psi_{N,0}\in L^2_{\mathrm{sym}}(\bR^{3N}),
    \qquad
    \|\Psi_{N,0}\|_2=1.
\]
The evolved quantum state then clearly satisfies $\Psi_{N,t}\in L^2_{\mathrm{sym}}(\bR^{3N})$ and $\|\Psi_{N,t}\|_2=1$, for any $t\in\mathbb{R}$. We define the associated \emph{$1$-particle reduced density matrix} $\gamma^{(1)}_{N,t}$ by its integral kernel
\begin{equation}\label{eq:rdm}
    \gamma^{(1)}_{N,t}(x_1,y_1)
    :=
    \int_{\bR^{3(N-1)}}\!\!
    \Psi_{N,t}(x_1,Z)
    \overline{\Psi_{N,t}(y_1,Z)}\,\d Z .
\end{equation}
Since $\|\Psi_{N,t}\|_2=1$, the trace of $\gamma^{(1)}_{N,t}$ is also normalized, that is
\begin{equation}\label{eq:tr-gamma-normalization}
    \Tr \gamma^{(1)}_{N,t}=1.
\end{equation}
The one-particle reduced density matrix $\gamma^{(1)}_{N,t}$ is a trace class operator which  describes the quantum state of a typical particle in the many-body system.

\medskip

For factorized initial data
\[
    \Psi_{N,0}=u_0^{\otimes N},
\]
with $\|u_0\|_2=1$, one expects the convergence
\[
    \gamma^{(1)}_{N,t}
    \longrightarrow
    |u_t\rangle\langle u_t| \quad \text{as } N\to\infty,
\]
in a suitable sense, where $u_t$ is a solution to the corresponding effective equation 
\begin{equation}\label{eq:BSE}
    \i \partial_t u_t
    =
    \sqrt{-\Delta+1}\,u_t
    +
    (w*|u_t|^2)u_t, \quad u_t\vert_{t=0}=u_0.
\end{equation}
Note that it is not difficult to verify that $\|u_t\|_2=1$ for all $t\in\bR$ since $w$ is real-valued.
We will refer to \eqref{eq:BSE} as the one-body boson star equation. Inspired by the work \cite{BFG2025} that deals with the one-body boson star equation \eqref{eq:BSE}, this paper is devoted to proving similar propagation estimates for the many-body boson star equation \eqref{eq:NbdQSE}. 
\subsection{Maximal and Minimal Velocity Bounds}\label{subsec:maxminvel}

In this paper, we prove maximal and minimal velocity bounds for \eqref{eq:NbdQSE}.
Roughly speaking, a maximal velocity bound asserts that a state initially localized in a region $X$ cannot instantaneously develop a significant mass in a distant region $Y$. More precisely, in Theorem \ref{thm:maxvel} below, we show that
if $\dist(X,Y)$ is large compared with the elapsed time $t$, then the mass observed in $Y$ is exponentially small. This can be seen as a Lieb--Robinson-type bound. 

On the other hand, a minimal velocity bound describes the complementary phenomenon: in a scattering regime, the position variable $x/t$ is asymptotically tied to the velocity operator associated with the kinetic energy.
As we show in Theorem \ref{thm:intro-many-body-minvel}\eqref{thm:minimal-scattering}, this implies that a suitably chosen state develops a significant mass at spatial infinity if the elapsed time is sufficiently large.
For the semi-relativistic dispersion relation
\begin{equation}\label{def:semi-rel-dispersion}
    \omega(\xi)=\sqrt{|\xi|^2+1},
\end{equation}
the corresponding velocity observable is
\begin{equation}\label{def:Theta}
    \Theta := \nabla_\xi \omega(-\i\nabla )
    = \frac{-\i\nabla}{\sqrt{-\Delta+1}}.
\end{equation}
Thus, at large times, one expects the spatial observable $x/t$ and the velocity observable~$\Theta$ to be close in a suitable sense. Minimal velocity estimates quantify this relation by showing that portions of the solution whose asymptotic position and velocity are incompatible have vanishing mass. In this sense, the term minimal velocity refers to a propagation estimate excluding components whose spatial escape rate is smaller than, or incompatible with, the velocity selected by the operator $\Theta$. We prove such an estimate in Theorem \ref{thm:intro-many-body-minvel} below.

In the present work, we first establish maximal and minimal velocity estimates for~\eqref{eq:NbdQSE}.
We then transfer these estimates to the one-body boson star equation \eqref{eq:BSE} by combining them with quantitative mean-field convergence estimates for the one-particle reduced density matrix.

\subsection{Main Results}\label{subsec:main}
In this section, we state the main results of the paper. 
The following hypothesis will be used in several of our results.

\begin{hypothesis}\label{hyp}
Let $u_0\in L^2(\bR^3)$ with $\|u_0\|_2=1$ and let $w\in L^1(\bR^3)\cap L^\infty(\bR^3)$ be an even function such that \eqref{eq:BSE} admits a global solution 
\begin{equation}
    u_t\in C\left([0,\infty), H^{1/2}(\bR^3)\cap L^\infty(\mathbb{R}^3)\right)\cap C^1\left([0,\infty), H^{-1/2}(\bR^3)\right)
\end{equation}
satisfying, for all $t\ge0$,
\begin{equation}\label{eq:decay}
    \|u_t\|_\infty\le C\langle t\rangle^{-3/2}.
\end{equation}
\end{hypothesis}
In \cite[Theorem 2.10]{BFG2025}, it is proven that Hypothesis \ref{hyp} is satisfied provided that $u_0$ belongs to a suitable Sobolev space and the norm $\|w\|_1$ is not too large (see Remark \ref{rk:hypotheses} below for more details). Note that the constant $C$ in \eqref{eq:decay} then depends on $w$ and $u_0$. \\

\noindent {\bf Many-body boson star equation. }
Our first result proves that the one-particle reduced density matrix $\gamma^{(1)}_{N,t}$ that is associated to $\Psi_{N,t}$ is close to the orthogonal projection onto the solution $u_t$ to \eqref{eq:BSE} in a quantitative way, uniformly in time. We provide a proof of this result in Section \ref{sec:mean-field} below.
\begin{theorem}[Mean-field convergence]\label{thm:intro-many-body-mf}
Under the assumptions of Hypothesis \ref{hyp}, assume that $\Psi_{N,t}$ solves \eqref{eq:NbdQSE} with initial datum $\Psi_{N,0}=u_0^{\otimes N}$. Let $\gamma^{(1)}_{N,t}$ be the one-particle reduced density matrix associated to $\Psi_{N,t}$ (see \eqref{eq:rdm}). Then, there exists $C>0$ such that for all $t\ge0$ and $N\in\mathbb{N}$, it holds
\begin{equation}
    \Tr\left|
        \gamma^{(1)}_{N,t}-|u_t\rangle\langle u_t|
    \right|
    \leq
    C N^{-1/2}.
\end{equation}
\end{theorem}

The following result is a Lieb--Robinson-type bound for the many-body boson star equation, which will be proved in Section \ref{sec:manybody-max-vel} below.

\begin{theorem}[Many-body maximal velocity estimate]\label{thm:maxvel}
Assume that $w:\bR^3\to\bR$ is even and such that
\begin{equation}
    w=w_1+w_2,\quad \text{where } w_1 \in L^1_{\mathrm{loc}}(\bR^3)\,, \,w_1\ge-\frac{\omega}{|x|} \text{ and } w_2\in L^{\infty}(\bR^3)\,,
\end{equation}
for some $\omega < \omega_c$, where $\omega_c>0$ is a critical constant defined in Remark \ref{rk:maximal}.\eqref{rk:critical}.
Let $X,Y$ be convex subsets of $\bR^3$. 
Let $\Psi_{N,t}$ be the solution to \eqref{eq:NbdQSE} with initial data $\Psi_{N,0}\in L^2(\mathbb{R}^{3N})$ with $\supp \Psi_{N,0} \subset X^N$ and $\| \Psi_{N,0} \|_2=1$.
Let $\gamma^{(1)}_{N,t}$ be the one-particle reduced density matrix associated to $\Psi_{N,t}$, as defined in \eqref{eq:rdm}. Then for all $N\in\bN$ and $t\ge0$, 
\begin{equation}\label{eq:maxvel}
    \Tr \big(\Id_{Y} \gamma^{(1)}_{N,t}\big)  \leq \exp\left(2(t-\dist(X,Y))\right) \,.
\end{equation}
\end{theorem}

\begin{remark}\label{rk:maximal}
We remark the following.
    \begin{enumerate}
        \item Note that there is nothing to prove in the case $\dist(X,Y)=0$. Also note that we do \emph{not} make any factorisation assumption on the initial data $\Psi_{N,0}$ in Theorem~\ref{thm:maxvel}. The maximal velocity estimate of Theorem \ref{thm:maxvel} is `sharp' in the sense that the right-hand side of \eqref{eq:maxvel} cannot be replaced by $C\exp\left(2(t-\dist(X,Y))\right)$ or $\exp\left(2C(t-\dist(X,Y))\right)$ for some $C<1$. This was shown in \cite[Remark 1.3]{BFG2025a} in the case $N=1$.
        \item At time $t=0$ and for disjoint sets $X$ and $Y$, estimate \eqref{eq:maxvel} reads $$0=\Tr \big(\Id_{Y} \gamma^{(1)}_{N,t}\big)  \leq \exp\left(2(t-\dist(X,Y))\right)=1,$$ which is obviously not sharp. Hence, our result, as stated, is not optimal for small times, but we actually prove a stronger result in the proof of Theorem \ref{thm:maxvel}: we show that for any unit vector $n\in\bR^3$ it holds
    \begin{equation}
        \Tr \big(\e^{-2n\cdot x} 
        \gamma^{(1)}_{N,t}\big)  \leq \e^{2t} \, \Tr \big(\e^{-2n\cdot x} \gamma^{(1)}_{N,0}\big) ,
    \end{equation}
which is actually sharp at $t=0$. In Theorem \ref{thm:maxvel}, we choose to state a slightly weaker result to make the speed of propagation more visible and the comparison with the literature simpler.
\item \label{rk:critical} The constant $\omega_c$ from the assumptions on $w$ in Theorem \ref{thm:maxvel} is defined as the optimal constant in the following functional inequality: for $\|u\|_2=1$,
\begin{equation}\label{eq:1bodystab}
    \int_{\bR^3}|\sqrt[4]{-\Delta} u|^2-\omega\int_{\bR^3}\int_{\bR^3}\frac{|u(x)|^2|u(y)|^2}{|x-y|}\,\d x \d y\ge 0.
\end{equation}
    By scaling, if $\omega>\omega_c$, then the left-hand side of \eqref{eq:1bodystab} is \emph{not} bounded from below for $\|u\|_2=1$, and thus by considering $u^{\otimes N}$, we deduce that $H_N$ is not bounded from below either. The fact that $\omega<\omega_c$ suffices to ensure lower-boundedness of~$H_N$ for symmetric functions was proved in \cite[Theorem 2]{LiebYau1987}.
    \item Theorem \ref{thm:maxvel} also holds under the following set of assumptions on $w$:\\
Assume that $w:\bR^3\to\bR$ is even and such that 
\begin{equation}
    w=w_1+w_2, \text{ where } w_1 \in L^1_{\mathrm{loc}}, \|w_{1,-}\|_{L^{3,\infty}}\leq \tilde{\omega}_\mathrm{c} \text{ and } w_2\in L^{\infty}(\bR^3)\,,
\end{equation}
where $w_{1,-}\in(0,\infty)$ is the negative part of $w_1$ and the critical constant $\tilde{\omega}_{\mathrm{c}}$ is defined through 
\begin{equation}\label{eq:K}
    \frac{1}{\tilde{\omega}_{\mathrm{c}}}:=\sup_{\substack{\tilde{w}\in L^{3,\infty}\setminus \{0\}\\ \psi \in \dot{H}^{1/2}(\mathbb{R}^3)\setminus \{0\}} } \frac{\big|\int_{\mathbb{R}^3} \tilde{w}|\psi|^2\big|}{\|\tilde{w}\|_{L^{3,\infty} }\|\psi\|^2_{\dot{H}^{1/2}}} \,.
\end{equation}
With this alternate set of assumptions, the only additional things to be shown with respect to what is written in the paper are, first that the Hamiltonian \eqref{eq:Nbody-Hamiltonian} is actually well defined, which is a consequence of the KLMN Theorem and, second, that $\tilde{\omega}_{\mathrm{c}}\in(0,\infty)$, which is a consequence of Hölder inequality and the Sobolev inequality in Lorentz spaces.
    \end{enumerate}
\end{remark}

The next result gathers several statements related to the minimal velocity for the many-body boson star equation and  will be proved in Section \ref{sec:manybody-min-vel} below. Roughly speaking, we show that for large $N$ and $1\ll t\ll \sqrt{N}$ the quantity $\frac{x_1}{t}\Psi_{N,t}$ is roughly equal to $\Theta_1\Psi_{N,t}$, where $\Theta_1$ is the action on the first particle of the operator $\Theta$ defined in \eqref{def:Theta}. 
\begin{theorem}[Many-body minimal velocity estimate]\label{thm:intro-many-body-minvel}
Assume Hypothesis \ref{hyp}, and assume furthermore that 
\begin{equation*}
    \|xu_0\|_2<\infty.
\end{equation*}
Assume that $\Psi_{N,t}$ solves \eqref{eq:NbdQSE} with initial datum $\Psi_{N,0}=u_0^{\otimes N}$. Then the following statements hold. 
\begin{enumerate}
\item \label{thm:minimal-L2bound}
A virial-type estimate holds: there exists $C>0$ depending on $w$ and $u_0$ such that, for all $t>0$ and $N\in\mathbb{N}$
, $N\geq2$,
    \begin{equation}\label{eq:virial-bound-Psi}
    \left\|
        \left(\frac{x_1}{t}-\Theta_1\right)
        \Psi_{N,t}
    \right\|_2^2
\leq
    C\left(\frac{1}{t^2}+
    \frac{1}{t} 
        +
        \frac{t^2}{N}
    \right).
\end{equation}
\item \label{thm:minimal-trace-bound}
Let $f,g\in C_c^\infty(\bR)$ satisfy
\begin{equation*}
    \supp f\cap\supp g=\emptyset.
\end{equation*}
Then there exists $C>0$ depending on $f$, $g$, $w$ and $u_0$ such that, for all $t>0$ and~$N\in\mathbb{N}$,
\begin{equation*}
     \Tr\left[
        g(\Theta^2) \,
        f^2\left(\frac{|x|^2}{t^2}\right)
        g(\Theta^2) \,
        \gamma^{(1)}_{N,t}
    \right]\le C\left(\frac{1}{t^2}
        +\frac{1}{t} 
        +
        \frac{t^2}{N}
    \right).
\end{equation*}

\item \label{thm:minimal-nu-bound}
Assume in addition that the solution $u_t$ to the boson star equation \eqref{eq:BSE} satisfies the following weighted dispersive bounds
\begin{equation}\label{eq:weighted1bodydecay}
    \|xu_t\|_\infty\lesssim 1,
    \qquad
    \|xu_t\|_2\|u_t\|_\infty\lesssim 1,
\end{equation}
uniformly in $t\ge0$. Define the weighted `bad-particle' component
\begin{equation*}
\nu_N(t):=\Tr\left(\left(\Id-|u_t\rangle\langle u_t|\right)x^2\left(\Id-|u_t\rangle\langle u_t|\right)\gamma_{N,t}^{(1)}\right)    .
\end{equation*}
Then there exists $C>0$ independent of $N\in\mathbb{N}$ and $t\ge0$, such that
\begin{equation}\label{eq:bound-nuN}
    \nu_N(t)
    \leq
    Ct^2.
\end{equation}

\item \label{thm:minimal-scattering}
Assume in addition that the solution $u_t$ to the boson star equation \eqref{eq:BSE} scatters in $L^2(\bR^3)$, i.e., assume that there exists $u_+\in L^2(\bR^3)$
such that
\begin{equation}\label{eq:hyp-scattering}
    \lim_{t\to\infty}
    \left\|
        u_t
        -
        \e^{-\i t\sqrt{-\Delta + 1}}u_+
    \right\|_2
    =0.
\end{equation}
Let $v\in(0,1)$ and assume that
\begin{equation}\label{eq:outgoing-assumption}
    \left\|
        \Id_{\{\Theta^2\ge v^2+\varepsilon\}}
        u_+
    \right\|_2
    =1
\end{equation}
for some $\varepsilon\in(0,1-v^2)$. Then
\begin{equation*}
    \lim_{\min\{t, N,  N/t^2\}\to \infty}
    \left\|
        \Id_{\{|x_1|\ge vt\}}
        \Psi_{N,t}
    \right\|_2^2
    =1.
\end{equation*}
\end{enumerate}
\end{theorem}

\begin{remark}\label{rk:hypotheses}
We make the following remarks about Hypothesis \ref{hyp} and Theorems \ref{thm:intro-many-body-mf}--\ref{thm:intro-many-body-minvel}. 
    \begin{enumerate}
    
    \item The estimates from Hypothesis \ref{hyp} are satisfied under suitable smallness and regularity assumptions on $u_0$:

    \begin{itemize}
        \item 
        For potentials $w \in L^1(\mathbb{R}^3) \cap L^\infty(\mathbb{R}^3)$, the global solution satisfies the linear decay rate 
        \begin{equation}
            \|u_t\|_\infty \lesssim \langle t \rangle^{-3/2}
        \end{equation}
         as shown in \cite{ChoOzawa2006,ChoOzawa2007}, provided the initial data $u_0$ is sufficiently small in the weighted Sobolev space $H^s(\mathbb{R}^3) \cap L^1(\mathbb{R}^3)$ (or $H^s(\mathbb{R}^3)  \cap L^2(\mathbb{R}^3,|x|^2 \d x)$) for $s$ sufficiently large. Note that in  Hypothesis \ref{hyp}, we impose $\|u_0\|_2=1$, so the smallness assumption on $u_0$ is equivalent, by scaling, to imposing that the potential $w$ is small. (For the detail, see e.g. \cite[Section 2.1]{Lee2019}.)
        \item The case $w \in L^1(\bR^3)$ was treated in \cite[Theorem 2.10]{BFG2025}. Let $s \ge 5/2$ and~$u_0$ in~$H^{s}(\bR^3)\cap H^{s, 1}(\bR^3)$. 
        Here
        \[
            \|f\|_{H^{s,1}} := \| (1 - \Delta)^{s/2} f \|_{1}.
        \]
        If $\|w \|_{1}\|u_0\|_{H^{s}\cap H^{s, 1}}$ is sufficiently small, then the solution $u_t$ of \eqref{eq:BSE} with initial data $u_0$ satisfies the inequalities
        \begin{equation*}
             \|u_t\|_{\infty} \lesssim \langle t \rangle^{-3/2} \| u_0\|_{H^s\cap H^{s,1}} \quad \|u_t\|_{H^s}\lesssim \|u_0\|_{H^s\cap H^{s,1}}.
        \end{equation*}
    \end{itemize}

    \item 
    
    The weighted dispersive bounds for the solution to the boson star equation \eqref{eq:BSE}
    \begin{equation}
        \|xu_t\|_\infty\lesssim 1,
        \qquad
        \|xu_t\|_2\|u_t\|_\infty\lesssim 1,
    \end{equation}
    used in Theorem \ref{thm:intro-many-body-minvel}\eqref{thm:minimal-nu-bound}, are for instance satisfied in the case $w \in L^1(\bR^3)$ if $u_0   \in H^{s, 1}(\bR^3)$, $\langle x \rangle u_0 \in H^s(\bR^3)$ with $s \ge 5/2$ and $\|w \|_{1}\|u_0\|_{H^{s}\cap H^{s, 1}}$ is sufficiently small. These results are proved in Lemma \ref{lemma:xu_t L infty} and \cite[Lemma 8.1]{BFG2025} respectively.
    
    \item The main ingredient to make our constants in the theorems independent of $t$ is the inequality
      \begin{equation}
                \|u_t\|_\infty \lesssim \langle t \rangle^{-3/2}\in L^1(0,\infty).
            \end{equation}
We remark that any integrable time decay rate, namely              \begin{equation}
                \|u_t\|_\infty \lesssim \langle t \rangle^{-r}
            \end{equation}
            for $r>1$, would be sufficient in order to have time-independent constants. Moreover, one can work with $w$ in suitable $L^p$ spaces, which allows to further generalise the results. We omit this for simplicity of the presentation.

    \item 
    \cite[Theorem 2.12]{BFG2025} provides sufficient conditions on $u_0$ and $w$ ensuring the existence of $u_+$ satisfying~\eqref{eq:hyp-scattering}.
    \end{enumerate}
\end{remark}

 \smallskip

\noindent {\bf One-body boson star equation. } 
In this section, we consider maximal and minimal velocity properties for the one-body boson star equation, which can also be found in~\cite{BFG2025}. We re-prove these results, using the ones stated above on the many-body boson star equation, in Section \ref{sec:BSE-prop-est} below.

\begin{theorem}[One-body maximal velocity estimate {\cite[Theorem 2.4]{BFG2025}}]\label{thm:intro-BSE-maxvel}
Assume Hypothesis \ref{hyp} 
and $\|xu_0\|_{2}<\infty$.
Let $u_t$ solve \eqref{eq:BSE} with normalized initial datum $u_0$ satisfying 
\[
    \supp u_0\subset X,
\]
with $X\subset\mathbb{R}^3$ convex.
Then, for every convex set $Y\subset\mathbb{R}^3$ and every $t\ge0$,
\[
\|\mathbbm{1}_Y u_t\|_2^2
    \leq
    \exp\big(2(t-\dist(X,Y))\big).
\]
\end{theorem}

\begin{theorem}[One-body minimal velocity estimate {\cite[Theorem 2.16]{BFG2025}}]\label{thm:intro-BSE-minvel}
Assume Hypothesis \ref{hyp}.
Let $u_t$ solve \eqref{eq:BSE}. Let $f,g\in C_c^\infty(\bR)$ satisfy
\[
    \supp f\cap\supp g=\emptyset.
\]
Then
\[
    \lim_{t\to \infty}\left\|
        f\left(\frac{|x|^2}{t^2}\right)
        g(\Theta^2)u_t
    \right\|_2^2
    =0.
\]
\end{theorem}

\subsection{Improved many-body minimal velocity estimate}
\label{subsec:many-body-minvel-improved}

We now use the one-body estimate Theorem \ref{thm:intro-BSE-minvel} and \eqref{eq:alphaNest} to improve the many-body minimal velocity estimate.

\begin{corollary}[Improved many-body minimal velocity estimate]
\label{cor:many-body-minvel-improved}
Under the assumptions of Theorem~\ref{thm:intro-many-body-minvel}\eqref{thm:minimal-trace-bound}, there exists $C>0$ depending on $f$, $g$, $w$ and $u_0$ such that, for all $t>0$ and $N\ge2$,
\begin{equation}\label{eq:many-body-minvel-improved}
    \Tr\left(g(\Theta^2)f^2\left(\frac{|x|^2}{t^2}\right)g(\Theta^2)\gamma^{(1)}_{N,t}\right)\le C\left(\frac{1}{t^2}+\frac{1}{t}+\frac{1}{N}\right).
\end{equation}
If, in addition, the assumptions of Theorem~\ref{thm:intro-many-body-minvel}\eqref{thm:minimal-scattering} hold, then
\begin{equation}\label{eq:many-body-scattering-improved}
    \lim_{\substack{\min(t,N)}\to\infty}\left\|\Id_{\{|x_1|\ge vt\}}\Psi_{N,t}\right\|_2^2=1.
\end{equation}
No relation between $t$ and $N$ is required in this limit.
\end{corollary}

\begin{remark}
Under the same assumptions, the one-body moment estimate from Appendix~\ref{sec:one-body-minvel-direct} gives
\[
    \sup_{t\ge0}\|(x-t\Theta)u_t\|_2<\infty.
\]
Applying the functional calculus estimate used in Subsection~\ref{subsec:min-vel-est}, we obtain
\[
    \left\|f\left(\frac{|x|^2}{t^2}\right)g(\Theta^2)u_t\right\|_2^2\le\frac{C}{t^2}.
\]
Using $\p_1+\q_1=\Id$ and \eqref{eq:alphaNest} as above, we obtain
\[
    \Tr\left(g(\Theta^2)f^2\left(\frac{|x|^2}{t^2}\right)g(\Theta^2)\gamma^{(1)}_{N,t}\right)\le C\left(\frac{1}{t^2}+\frac{1}{N}\right).
\]
\end{remark}

\subsection{History and Comparison with Previous Works}\label{sec:hist}

The boson star equation \eqref{eq:BSE}
is the semi-relativistic Hartree equation arising as the effective dynamics of pseudo-relativistic bosons. The case of an attractive Newtonian interaction potential $w(x)=-\kappa |x|^{-1}$ in three dimensions, was considered in various works in the literature.  The static mean-field problem goes back to the work of Lieb and Yau \cite{LiebYau1987}, where the Chandrasekhar theory was obtained as a limit of a many-particle system in the mean-field regime with semi-relativistic kinetic energy. The time-dependent mean-field derivation of the boson star equation was proved by Elgart and Schlein \cite{ElgartSchlein2007}, and quantitative convergence rates were later obtained in \cite{Lee2013}. Our proof uses a counting argument in the spirit of Pickl's method \cite{Pickl2011}, which has been developed further to treat singular interactions and quantitative convergence rates \cite{KnowlesPickl2010,ChenLee2011,ChenLeeLee2018}. Since the estimates in the present paper are used for large times, we also rely on ideas related to mean-field convergence with improved time dependence and uniform-in-time bounds, as studied for instance in \cite{Lee2019,DietzeLee2022,NamRademacherSoffer2025}.

The analysis of the semi-relativistic Hartree equation has also been studied from the viewpoint of nonlinear dispersive equations. Global well-posedness in the energy space, including the critical Newtonian interaction, was proved in \cite{Lenzmann2007}, while low-regularity local well-posedness was obtained in \cite{HerrLenzmann2014}. Solitary waves and their orbital stability were studied in \cite{FrohlichJonssonLenzmann2007}. For short-range Hartree interactions, global existence and scattering were proved in \cite{ChoOzawa2006,ChoOzawa2007,ChoOzawaSasaki2008,HerrTesfahun2015}. 
In the long-range case, usual scattering has to be replaced by modified scattering; for kernels $w(x)=\lambda |x|^{-\alpha}$, $1\leq \alpha<3$, this was studied in \cite{Pusateri2014}.
We also recall the complementary variational literature on boson-star models, in particular on ground-state collapse and many-body blow-up profiles in the presence of external potentials \cite{Nguyen2017BSE,Nguyen2019ManyBodyBSE}.

Propagation estimates form another line of work which is important for the present paper. In many-body scattering theory, maximal and minimal velocity bounds are standard tools; see for instance \cite{SigalSoffer1987,SigalSoffer1990,Graf1990,Skibsted1991,HunzikerSigalSoffer1999}. In quantum many-body dynamics on lattices, the analogous principle of an effective light cone is expressed by Lieb--Robinson bounds, originally proved for quantum spin systems in \cite{LiebRobinson1972}; see also \cite{NachtergaeleSims2006,ChenLucasYin2023}. Extensions to bosonic systems are substantially more delicate because of unbounded local Hilbert spaces and unbounded local energies. Recent works have developed Lieb--Robinson and maximal-particle-transport bounds for Bose--Hubbard-type systems and related bosonic models \cite{FaupinLemmSigal2022,KuwaharaSaito2021,YinLucas2022,KuwaharaVuSaito2024}. Closest in spirit to the mean-field aspect of the present work is \cite{LemmRademacherZhang2024}, where ballistic propagation bounds for fluctuations around a condensate are used to obtain a local enhancement of the mean-field approximation on the lattice.
We also mention the recent Oberwolfach mini-workshop report \cite{GiacomelliLeeLemm2025}, which highlights current interactions between the mathematical analysis of quantum many-body systems in scaling regimes and the study of strongly correlated systems through Lieb--Robinson-type propagation bounds.

The aim of the present paper is different in spirit. We prove velocity localization estimates for the boson star equation from the underlying many-body dynamics. More precisely, we first establish suitable maximal and minimal velocity bounds for the many-body boson star equation, and then pass to the Hartree dynamics using quantitative mean-field convergence estimates based on Pickl's counting method. In this way, velocity localization for the effective equation is obtained as a consequence of velocity localization at the many-body level.

\subsection{Strategy of the Proof}\label{subsec:strategy}

We briefly explain the strategy of the proofs. More details are given in the corresponding sections below.

\subsubsection*{Proof strategy for Theorem \ref{thm:intro-many-body-mf} (Mean-field convergence)}
We use Pickl's counting functional
\[
    \alpha_N(t)
    =
    \|\q_1(t)\Psi_{N,t}\|_2^2,
\]
where $\q_1(t)$ is the projection on the orthogonal complement of $|u_t\rangle\langle u_t|$ acting on the first particle. The quantity $\alpha_N(t)$ measures the fraction of particles outside the condensate. We prove the estimate
\begin{equation}\label{eq:alpha_intro}
    \alpha_N(t)\le\frac{C}{N}
\end{equation}
using a Grönwall argument. Thanks to the inequality
\begin{equation}
    \|u_t\|_\infty\le C\langle t\rangle^{-3/2},
\end{equation}
which holds from Hypothesis \ref{hyp}, one can show that the constant in \eqref{eq:alpha_intro} is independent of $t$. Then Theorem \ref{thm:intro-many-body-mf} follows from the trace norm estimate (e.g. in \cite[Lemma 2.3]{KnowlesPickl2010})
\[
    \Tr\left|
        \gamma^{(1)}_{N,t}
        -
        |u_t\rangle\langle u_t|
    \right|
    \lesssim
    \alpha_N(t)^{1/2}.
\]

\subsubsection*{Proof strategy for Theorem \ref{thm:maxvel} (Many-body maximal velocity estimate)}
The proof follows along the lines of \cite[Theorem 2.4]{BFG2025}, but working in the setting of one-particle density matrices. It relies on approximations of the operator 
\begin{equation}
    \e^{\ell_\varepsilon(x)} \langle \nabla \rangle \e^{-\ell_\varepsilon(x)}-\e^{-\ell_\varepsilon(x)} \langle \nabla \rangle \e^{\ell_\varepsilon(x)}
\end{equation}
in operator norm uniformly in $\varepsilon$, where $\ell_\varepsilon(x)$ denotes a smoothened version with small parameter $\varepsilon>0$ of a signed distance to a hyperplane separating $X$ and $Y$.

\subsubsection*{Proof strategy for Theorem \ref{thm:intro-many-body-minvel} (Many-body minimal velocity estimate)}
We consider 
\begin{equation}
       \mathcal M_N(t)
    :=
    \langle \Psi_{N,t}, |x_1-t\Theta_1|^2\Psi_{N,t}\rangle.
\end{equation}
The quantity $\mathcal{M}_N(t)/t^{2}$ controls the mismatch between the \emph{asymptotic ballistic position variable} $x_1/t$ and the \emph{velocity observable} $\Theta_1$. Recall that $\Theta_1$ is the operator $\Theta$ defined in \eqref{def:Theta} acting on the first variable. 

Using the mean-field estimate from Theorem \ref{thm:intro-many-body-mf}, which controls the number of particles outside
the condensate, and an additional estimate for the interaction terms coupling
the condensate and its orthogonal complement, we obtain Theorem \ref{thm:intro-many-body-minvel}\eqref{thm:minimal-L2bound}.

Theorem \ref{thm:intro-many-body-minvel}\eqref{thm:minimal-trace-bound} is obtained from Theorem \ref{thm:intro-many-body-minvel}\eqref{thm:minimal-L2bound} following a similar approach as in \cite{BFG2025}.

The proof of Theorem \ref{thm:intro-many-body-minvel}\eqref{thm:minimal-nu-bound} is based on a Grönwall argument for $\nu_N(t)$ and uses Theorem \ref{thm:intro-many-body-mf}. Applying  Theorem \ref{thm:intro-many-body-mf} and \ref{thm:intro-many-body-minvel}\eqref{thm:minimal-nu-bound} we conclude by proving Theorem \ref{thm:intro-many-body-minvel}\eqref{thm:minimal-scattering}

\subsubsection*{Proof strategy for Theorem \ref{thm:intro-BSE-maxvel} and Theorem \ref{thm:intro-BSE-minvel} (One-body maximal and minimal velocity estimate)}
In Theorem \ref{thm:intro-BSE-maxvel} and Theorem \ref{thm:intro-BSE-minvel}, we transfer the $N$-body information from Theorem \ref{thm:maxvel} and Theorem \ref{thm:intro-many-body-minvel}
to the (one-body) boson star equation. Since the many-body estimates are formulated in terms of the one-particle reduced density matrix, the trace norm convergence from Theorem \ref{thm:intro-many-body-mf} allows us to pass to the limit $N\to\infty$. In this way, the maximal and minimal velocity estimates for the $N$-body system yield the corresponding propagation estimates, obtained in \cite{BFG2025}, for the solution $u_t$ of the boson star equation. 

\subsection{Notation}

For $d \in \bN$, we denote by $\bR^d$ the $d$-dimensional Euclidean space. For a measurable set $S \subseteq \bR^d$ and a function $f : S \to \bC$, we write
\[
    \|f\|_{L^p(S)} := \bigg( \int_S |f(x)|^p \, \d x \bigg)^{1/p},
    \qquad 1 \le p < \infty,
\]
with the usual modification for $p=\infty$. When $S = \bR^3$, we simply write $\|f\|_{p}$ to denote $\|f\|_{L^p(\bR^3)}$.
Given a measurable set $Y \subset \bR^3$, we denote by $\Id_Y$ the multiplication operator by the indicator function of~$Y$. For a bounded linear operator $O : L^2(S_1) \to L^2(S_2)$, we denote its operator norm by $\|O\|_{\op}$.
For a trace-class operator $\gamma$ on $L^2(\bR^3)$, we write $\Tr \gamma$ for its trace and $\|\gamma\|_{\Tr} := \Tr |\gamma|$ for its trace norm. The rank-one operator associated with $u \in L^2(\bR^3)$ is denoted by $|u\rangle\langle u|$.

Throughout the paper, the letter $C$ denotes a positive constant whose value may change from line to line. The dependence on parameters is stated explicitly when relevant. The symbol $a\lesssim b$ denotes the inequality $a\le C \cdot b$ for some positive constant $C$. We also use the notation $\langle a \rangle = (1+|a|^2)^{1/2}, a \in \bR^n$. 

\subsection{Organization of the Paper}\label{subsec:organization}
The paper is organized as follows: 
in Section~\ref{sec:mean-field}, we use Pickl's counting functional for the proportion of particles outside the condensate to prove estimates on the convergence of the one-particle reduced density matrix to the projection onto the solution to \eqref{eq:BSE} uniformly in time. 
In Section~\ref{sec:manybody-max-vel}, we provide the proof of the many-body maximal velocity estimate.
In Section~\ref{sec:manybody-min-vel}, we prove the many-body minimal velocity estimate. The main ingredient is a moment bound, which controls the mismatch between the asymptotic position variable $x_1/t$ and the velocity observable $\Theta_1$ on mean-field time scales.
Finally, in Section~\ref{sec:BSE-prop-est}, we transfer the many-body maximal and minimal velocity estimates to the effective boson star equation by using the trace norm convergence of the one-particle reduced density matrix. We also prove the improved many-body minimal velocity estimate. 
Appendix \ref{sec:one-body-minvel-direct} contains an alternative one-body proof of the minimal velocity estimate, 
while in Appendix \ref{app: t decay} we prove that the assumption of Theorem \ref{thm:intro-many-body-minvel}\eqref{thm:minimal-nu-bound} is satisfied for suitable $w$ and $u_0$. 

\subsection*{Acknowledgments}
We would like to thank Phan Thành Nam for helpful comments. This research was funded, in whole or in part, by l’Agence Nationale de la Recherche (ANR),
project ANR-22-CE92-0013.

C. Dietze would like to thank Jacques Smulevici for a helpful discussion. 

J. Lee was partially supported by the Global-Learning \& Academic research institution for Master's-PhD students, and Postdocs (G-LAMP) Program of the National Research Foundation of Korea (NRF) grant funded by the Ministry of Education (No. RS-2025-25442355), the NRF grant funded by the Korea government (MSIT) (No. RS-2026-25589357), and a grant from Kyung Hee University in 2026 (KHU-20262263).

\subsection*{AI declaration}
For the present work, we used AI-tools 
GPT-6.1 Sol
to correct grammar and spelling, and to modify Figure \ref{fig:propagation}.

\vspace{3em}

\section{Mean-Field Approximation}\label{sec:mean-field}

In this section we prove the mean-field approximation for the $N$-body dynamics \eqref{eq:NbdQSE} for all times $t\ge 0$, namely we prove Theorem \ref{thm:intro-many-body-mf}. The main quantity is the number of particles outside the condensate, measured by a projection-counting functional $\alpha_N(t)$ associated with the condensate wave function $u_t$.

\subsection{Counting Functional and Convergence of Reduced Densities}\label{subsec:alphaN}

Let $u_t$ be the solution to the one-body boson star equation \eqref{eq:BSE}. We define the time-dependent orthogonal projectors acting on the one-particle Hilbert space $L^2(\bR^3)$ by
\begin{equation}\label{eq:projectors}
    \p_t^u := |u_t\rangle\langle u_t|,
    \qquad
    \q_t^u := 1-\p_t^u.
\end{equation}
These operators are lifted to the $N$-particle Hilbert space $L^2(\bR^{3N})$ by acting on the $j$-th particle. We denote the corresponding operators by
\begin{equation}\label{eq:pqdef}
    \p_j := (\p_t^u)_j,
    \qquad
    \q_j := (\q_t^u)_j.
\end{equation}
Since $\|u_t\|_2=1$ for every $t$ (because $\|u_0\|_2=1$ by assumption and the mass is conserved, see \cite[Lemma 4.2]{Lenzmann2007}) then $\p_j$ and $\q_j$ are orthogonal projectors satisfying
\[
    \p_j^2=\p_j,
    \qquad
    \q_j^2=\q_j,
    \qquad
    \p_j\q_j=0.
\]
We define the projection-counting operator
\begin{equation}\label{def:nhatu}
    \hat n^u
    :=
    \frac1N \sum_{j=1}^{N}\q_j,
\end{equation}
and the projection-counting functional
\begin{equation}\label{def:alphaN}
    \alpha_N(t)
    :=
    \langle \Psi_{N,t},\widehat n^u\Psi_{N,t}\rangle
    :=
    \frac1N\sum_{j=1}^N
    \langle \Psi_{N,t},\q_j\Psi_{N,t}\rangle .
\end{equation}
For bosonic wave functions, symmetry gives the simpler expression
\begin{equation}\label{eq:alpha_symmetry}
    \alpha_N(t)
    =
    \|\q_1\Psi_{N,t}\|_2^2.
\end{equation}
Thus $\alpha_N(t)$ measures the fraction of particles outside the condensate $u_t$.

The convergence of the reduced density matrix follows from the convergence of $\alpha_N(t)$. More precisely, the trace norm distance is controlled by the square root of the counting functional in the following way
\begin{equation}\label{eq:trace-dist-alpha}
    \Tr\left|
        \gamma^{(1)}_{N,t}
        -
        |u_t\rangle\langle u_t|
    \right|
    \leq
    \sqrt{8\alpha_N(t)}.
\end{equation}
We recall how to obtain inequality \eqref{eq:trace-dist-alpha} (see e.g. \cite[Lemma 2.3]{KnowlesPickl2010}). Since $\gamma^{(1)}_{N,t}$ is a non-negative density matrix and $|u_t\rangle\langle u_t|$ is a rank-one projector, we have
\[
    \Tr\left|
        \gamma^{(1)}_{N,t}
        -
        |u_t\rangle\langle u_t|
    \right|
    =
    2\left\|
        \gamma^{(1)}_{N,t}
        -
        |u_t\rangle\langle u_t|
    \right\|_\op
    \le
    2\left\|
        \gamma^{(1)}_{N,t}
        -
        |u_t\rangle\langle u_t|
    \right\|_{\mathrm{HS}},
\]
where $\|\cdot\|_{\mathrm{HS}}$ is the Hilbert-Schmidt norm. One can then compute
\[
    \left\|
        \gamma^{(1)}_{N,t}
        -
        |u_t\rangle\langle u_t|
    \right\|_{\mathrm{HS}}^2
    =
    \Tr\left(\gamma^{(1)}_{N,t}
        -
        |u_t\rangle\langle u_t|\right)^2
    =
    2-2\langle u_t, \gamma^{(1)}_{N,t} u_t \rangle
    =
    2\alpha_N(t),
\]
if $\|u_t\|_2= \|\Psi_{N,t}\|_2=1$ for all $t$.
Therefore, proving Theorem \ref{thm:intro-many-body-mf} reduces to proving
\begin{equation}\label{eq:alphaNest}
    \alpha_N(t)\le\frac{C}{N}
\end{equation}
for a constant $C>0$ that is independent of $t$ and $N$.

\subsection{Proof strategy for the estimate \eqref{eq:alphaNest} on \texorpdfstring{$\alpha_N(t)$}{alpha_N(t)}}

We briefly recall the standard estimate for $\alpha_N(t)$. Let $H_N$ be the microscopic Hamiltonian in \eqref{eq:Nbody-Hamiltonian}, and let
\begin{equation}
\label{def: h_j}
        h_j^{\MF}(t)
    :=
    \sqrt{-\Delta_j+1}
    +
    (w*|u_t|^2)(x_j)
\end{equation}
be the effective mean-field Hamiltonian acting on particle $j$. Differentiating $\alpha_N(t)$ and using the equations of motion for $\Psi_{N,t}$ and $u_t$, one obtains
\begin{align}
    \frac{\d}{\d t}\alpha_N(t)
    &=
    \i\langle \Psi_{N,t},[H_N,\q_1]\Psi_{N,t}\rangle
    +
    \langle \Psi_{N,t},\partial_t\q_1\Psi_{N,t}\rangle
    \nonumber \\
    &=
    \i\bigg\langle
        \Psi_{N,t},
        \Big[
            H_N-\sum_{j=1}^N h_j^{\MF}(t),
            \q_1
        \Big]
        \Psi_{N,t}
    \bigg\rangle .
    \label{comp:alpha-prime}
\end{align}
Here the kinetic terms cancel, and only the difference between the microscopic interaction and the mean-field interaction remains.

The estimate is then obtained by inserting the decomposition
\[
    1=\p_j+\q_j
\]
around the interaction terms and classifying the resulting contributions according to the number of $\q$-projectors. The leading condensate contribution cancels with the mean-field potential. Formally, this cancellation is reflected in the identity
\[
    \p_2 \, w(x_1-x_2) \, \p_2
    =
    (w*|u_t|^2)(x_1) \, \p_2.
\]
The remaining terms contain at least one $\q$-projector and are controlled by $\alpha_N(t)$ and the mean-field scaling.

A Gr\"onwall estimate yields
\begin{equation}\label{eq:alpha-gronwall}
    \left|
        \frac{\d}{\d t}\alpha_N(t)
    \right|
    \leq
    C(t)\left(\alpha_N(t)+\frac1N\right),
\end{equation}
where $C(t)$ is integrable in time under the decay estimate on $u_t$ in \eqref{eq:decay}. Hence,
\begin{equation}\label{eq:alpha-bound}
    \alpha_N(t)
    \leq
    \left(\alpha_N(0)+\frac1N\right)
    \exp\left(\int_0^t C(s)\,\d s\right)
    =\frac1N\exp\left(\int_0^t C(s)\,\d s\right)
    ,
\end{equation}
where we used in the last step that the initial data is factorized, namely $\Psi_{N,0}=u_0^{\otimes N}$, so that $\alpha_N(0)=0$. In our setting, we can control the constant $C(s)$ such that 
\begin{equation}
    \int_0^\infty C(s)\,\d s<\infty,
\end{equation}
which gives a time-independent constant in \eqref{eq:alphaNest}. 

\subsection{Proof of \eqref{eq:alphaNest}}
We start by giving an expression for $\frac{\d}{\d t}\alpha_N(t)$. 

\begin{lemma}
\label{lemma: expression alpha'}With the assumptions and notations of Theorem~\ref{thm:intro-many-body-mf} and $\alpha_N(t)$ defined as in~\eqref{def:alphaN}, it holds that
\begin{equation*}
   \frac{\d}{\d t} \alpha_N(t)= -2 \Im \langle \Psi_{N,t} , \p_1 V_t(x_1, x_2) \q_1 \Psi_{N,t}\rangle
\end{equation*}
where
\begin{equation}\label{def:Vt}
    V_t(x_1, x_2) := w(x_1 - x_2) -w*|u_t|^2(x_1).
\end{equation}
\end{lemma}
\begin{proof}
Since $u_t$ solves the Hartree equation
\[
    \i\partial_t u_t=h^{\MF}(t)u_t,
    \qquad
    h^{\MF}(t)=\sqrt{-\Delta+1}+(w*|u_t|^2),
\]
the time-dependent projector $\q_1$ satisfies
\[
    \partial_t\q_1=-\i[h^{\MF}_1(t),\q_1],
\]
with $h^{\MF}_1(t)$ defined in \eqref{def: h_j}. 
Hence, differentiating $\alpha_N$, we obtain \eqref{comp:alpha-prime}. 
Then, since $\q_1$ only acts on the first particle and since we apply the operators to symmetric functions
\begin{align}
    [ H_N - h^{\MF}_1 , \q_1 ] & = \, \Big[\frac{1}{N-1}\sum_{k=2}^Nw(x_1 - x_k) -  w*|u_t|^2(x_1) ,\,  \q_1  \Big]\nonumber\\
    & =  \Big[w(x_1 - x_2) -w*|u_t|^2(x_1)  , \q_1\Big]. \label{comp: [H - h, q1]}
\end{align}
Given the definition \eqref{def:Vt}, we therefore get
\begin{align*}
    \frac{\d}{\d t}\alpha_N(t) &= \i \langle \Psi_{N,t}, [V_t(x_1, x_2), \q_1] \Psi_{N,t} \rangle =-2\Im \langle \Psi_{N,t},V_t(x_1, x_2) \q_1 \Psi_{N,t}\rangle \\
    &= -2\Im \langle \Psi_{N,t},\p_1V_t(x_1, x_2) \q_1 \Psi_{N,t}\rangle,
\end{align*}
where we used the relation $1=\p_1 + \q_1$. 
\end{proof}

As explained above, proving the mean-field convergence reduces to bounding $\alpha_N(t)$, which we do in the following proposition. The proof is analogous to \cite[Lemma 3.2]{Pickl2011}, we nevertheless include it for completeness. 

\begin{proposition}
\label{lemma: bound alpha}
Assume that the assumptions of Theorem~\ref{thm:intro-many-body-mf} are satisfied. Let~$\alpha_N(t)$ be defined as in~\eqref{def:alphaN}, and
\begin{equation}\label{eq:defCs}
    C(s):=8\|u_s\|_{\infty}\|w\|_{2}.
\end{equation}
Then
\begin{equation}
     \alpha_N(t) \le\frac{1}{N}\exp\left(\int_0^t C(s)\, \d s\right).
\end{equation}
\end{proposition}

\begin{proof}
We prove the inequality 
\begin{equation}
\label{eq: bound alpha'}
    \frac{\d}{\d t}\alpha_N(t)\le C(t)\left(\frac{1}{N}+\alpha_N(t)  \right) ,
\end{equation}
the statement then follows applying Grönwall inequality to $f(t) = \alpha_N(t)  + N^{-1} $ and the property $\alpha_N(0)=0$, due to the factorization of the initial data $\Psi_{N,0} =u_0^{\otimes N}$. To obtain \eqref{eq: bound alpha'}, we use the expression given in Lemma \ref{lemma: expression alpha'} and insert $\p_2+ \q_2=1$
\begin{align*}
    \frac{\d}{\d t}\alpha_N(t) =& -2 \Im \langle \Psi_{N,t} , \p_1 V_t(x_1, x_2) \q_1 (\p_2 + \q_2)\Psi_{N,t}\rangle\\
    =& -2  \Im \langle \Psi_{N,t} , \p_1 (\p_2 + \q_2) V_t(x_1, x_2) \q_1 \p_2 \Psi_{N,t}\rangle\\
    & - 2 \Im \langle \Psi_{N,t} , \p_1 V_t(x_1, x_2) \q_1  \q_2\Psi_{N,t}\rangle.
\end{align*}
In this expression we first have 
\begin{equation*}
    \Im \langle \Psi_{N,t} , \p_1 \q_2 V_t(x_1, x_2) \q_1 \p_2 \Psi_{N,t}\rangle = 0,
\end{equation*}
since $ \p_1 \q_2 V_t(x_1, x_2) \q_1 \p_2$ is invariant by taking the adjoint and exchanging the first and second particles thanks to the symmetry of $w$ and of the wave function. 

Next, we remark that 
\begin{equation*}
    \p_2 w(x_1-x_ 2) \p_2= \braket{u_t|w(x_1-\cdot )} { u_t}  \p_2 = w*|u_t|(x_1) \p_2
\end{equation*}
and therefore $\p_2V_t(x_1, x_2) \p_2=0 $. This implies 
\begin{equation*}
    \Im \langle \Psi_{N,t} , \p_1 \p_2 V_t(x_1, x_2) \q_1 \p_2 \Psi_{N,t}\rangle = 0.
\end{equation*}
Note that $\hat{n}^u$, defined in \eqref{def:nhatu}, preserves $\mathrm{Ran}(q_1)$ and $q_1\hat{n}^uq_1\ge N^{-1}q_1$. Therefore $\hat{n}^u$ is invertible on $\mathrm{Ran}(q_1)\subset L^2(\bR^{3N})$ and we have
\begin{align}
    \frac{\d}{\d t}\alpha_N(t) = & -2  \Im \langle (\hat{n}^u)^{1/2} V_t(x_1, x_2)\p_1 \Psi_{N,t} , (\hat{n}^u)^{-1/2} \q_1  \q_2\Psi_{N,t}\rangle\nonumber\\
    \leq & \, 2 \|(\hat{n}^u)^{1/2} V_t(x_1, x_2)\p_1 \Psi_{N,t}\|_{2} \|(\hat{n}^u)^{-1/2} \q_1  \q_2\Psi_{N,t}\|_{2}. \label{comp:bound alpha'}
\end{align}
Since $[\q_j,\hat{n}^u] =0 $ and $\langle\Psi, \hat{n}^u\Psi\rangle\le\frac{N-1}{N}\langle \q_2\Psi,\q_1 \q_2 \Psi\rangle$ for $\Psi\in L^2_{\mathrm{sym}}(\bR^{3N})$, we first have 
\begin{align}
    \|(\hat{n}^u)^{-1/2} \q_1  \q_2\Psi_{N,t}\|_{2} ^2
    &=
    \langle (\hat{n}^u)^{-1/2}q_2\Psi_{N,t}, \q_1 (\hat{n}^u)^{-1/2} \q_2 \Psi_{N,t}\rangle \nonumber \\
    &\le
    \frac{N}{N-1}\langle \Psi_{N,t},   \q_2 \Psi_{N,t}\rangle 
    \le
    2\alpha_N(t).
\end{align}
Then, by definition of $\hat{n}^u$ and since $\q_1 + \q_2 \le 2 \,\mathrm{Id}$ as quadratic forms we obtain 
\begin{align}
    \|(\hat{n}^u)^{1/2} V_t(x_1, x_2)\p_1 \Psi_{N,t}\|_{2} ^2 
    =& \frac{1}{N}\langle  \Psi_{N,t}, \p_1  V_t(x_1, x_2)(\q_1 + \q_2)V_t(x_1, x_2) \p_1 \Psi_{N,t} \rangle \nonumber\\
    &+\frac{(N-2)}{N} \langle  \Psi_{N,t},  \p_1  V_t(x_1, x_2)\q_3 V_t(x_1, x_2) \p_1 \Psi_{N,t} \rangle\nonumber\\
    \le & \frac{2}{N} \langle  \Psi_{N,t}, \p_1  V_t(x_1, x_2)^2\p_1 \Psi_{N,t} \rangle \nonumber\\
    & + \langle \q_3\Psi_{N,t},  \p_1  V_t(x_1, x_2)^2 \p_1 \q_3\Psi_{N,t} \rangle\nonumber\\
    \le &\|\p_1  V_t(x_1, x_2)^2\p_1 \|_{\op} \left( \frac{2}{N} +\alpha_N(t) \right), \label{eq:bound sqrt n}
\end{align}
where in the last inequality we used the symmetry of $\Psi_{N,t}$ to write $\|\q_3 \Psi_{N,t} \|_{2}^2=\alpha_N(t)$. Note that 
\begin{equation}\label{eq:p1vt}
    \p_1 V_t(x_1,x_2)^2 \p_1
    =
    \ell_t(x_2)\p_1,
\end{equation}
where
\[
    \ell_t(x_2)
    =
    \int_{\mathbb R^3}
        |u_t(x_1)|^2
        |V_t(x_1,x_2)|^2
    \,\dd x_1 .
\]
Recall the definition of $V_t$ given in \eqref{def:Vt}.
Since $w\in L^1(\bR^3) \cap L^\infty (\bR^3)$ and $\|u_t\|_2=1$, we have
\begin{equation}
\label{comp: sup V x_2}
    \sup_{x_2}\|V_t(\cdot,x_2)\|_{2}
    \le
    \|w\|_{2}+\|w*|u_t|^2\|_{2}
    \le
    2\|w\|_{2} .
\end{equation}
Thus, we obtain by \eqref{eq:p1vt}
\begin{multline}
    \|\p_1  V_t(x_1, x_2)^2\p_1 \|_{\op}=\|\ell_t(x_2)\p_1\|_{\op}\le \|\ell_t\|_\infty 
    \le \|u_t\|_{\infty}^2 \sup_{x_2} \|V_t(\cdot, x_2)\|_{2}^2 \le 4\|u_t\|_{\infty}^2\|w\|_{2}^2. \label{eq:estimate-p-Vt2-p}
\end{multline}
Combining this with \eqref{comp:bound alpha'}-\eqref{eq:bound sqrt n} we obtain 
\begin{align*}
  \frac{\d}{\d t}   \alpha_N(t) \le& 4 \|u_t\|_{\infty}\|w\|_{2} \sqrt{2\alpha_N(t)\left( \frac{2}{N} + \alpha_N(t)\right)}\\
    \le & 8 \|u_t\|_{\infty}\|w\|_{2} \left(\frac{1}{N} +\alpha_N(t)\right).
\end{align*}
This implies \eqref{eq: bound alpha'} and thereby concludes the proof.
\end{proof}

\begin{proof}[Proof of Theorem \ref{thm:intro-many-body-mf}]
Combining \eqref{eq:trace-dist-alpha} and Proposition \ref{lemma: bound alpha}, we get
\begin{equation}
    \Tr\left|
        \gamma^{(1)}_{N,t}-|u_t\rangle\langle u_t|
    \right|\le \sqrt{\frac{8}{N}}\exp\left(\frac{1}{2}\int_0^t C(s)\, \d s\right),
\end{equation}
where $C(s)$ was defined in \eqref{eq:defCs}. By Hypothesis \ref{hyp}, we know that 
\begin{equation}
    \|u_t\|_{\infty}\lesssim\langle t\rangle^{-3/2},
\end{equation}
so that
\begin{equation}
    \int_0^\infty C(s)\, \d s<\infty.
\end{equation}
This proves 
\begin{equation}
    \Tr\left|
        \gamma^{(1)}_{N,t}-|u_t\rangle\langle u_t|
    \right|\le \frac{C}{\sqrt{N}},
\end{equation}
with a constant $C>0$ that is independent of $t$ and $N$, and this completes the proof of Theorem \ref{thm:intro-many-body-mf}.
\end{proof}

\vspace{1em}

\section{Many-Body Maximal Velocity Estimate}\label{sec:manybody-max-vel}

The goal of this section is to prove the many-body maximal velocity estimate stated in Theorem \ref{thm:maxvel}. Under suitable assumptions on the interaction potential $w$, we will prove that if the initial many-body wave function $\Psi_{N,0}\in L^2(\mathbb{R}^{3N})$ is supported in a convex set~$X$, and we consider another convex set $Y$ that is disjoint of $X$, then the proportion of the particles that are at time $t$ in the set $Y$ is bounded by $\e^{2(t-\dist(X,Y))}$. Namely, it holds
\begin{equation}\label{eq:maxvel1}
    \Tr \big(\Id_{Y} \gamma^{(1)}_{N,t}\big)  \leq \e^{2(t-\dist(X,Y))} \,,
\end{equation}
where $\gamma^{(1)}_{N,t}$ is the one-particle reduced density matrix associated to $\Psi_{N,t}$. In particular, this shows that if $t$ is much smaller than the distance of $X$ and $Y$, then the proportion of particles at time $t$ that are in $Y$ is exponentially small.  We underline the fact that unlike in the other sections of this paper, no factorisation assumption on $\Psi_{N,0}$ is needed in Theorem \ref{thm:maxvel}.

\smallskip

\subsection{Proof strategy for Theorem \ref{thm:maxvel}}
The proof of Theorem \ref{thm:maxvel} is an adaptation to one-particle reduced density matrices of the proof of Theorem~2.4 in \cite{BFG2025}, following a general strategy that originates from \cite{SigalWu}. Let us introduce some notations that will be useful for the proof of Theorem~\ref{thm:maxvel}, for some fixed convex sets $X$ and $Y$ with $\dist(X,Y)>0$. 

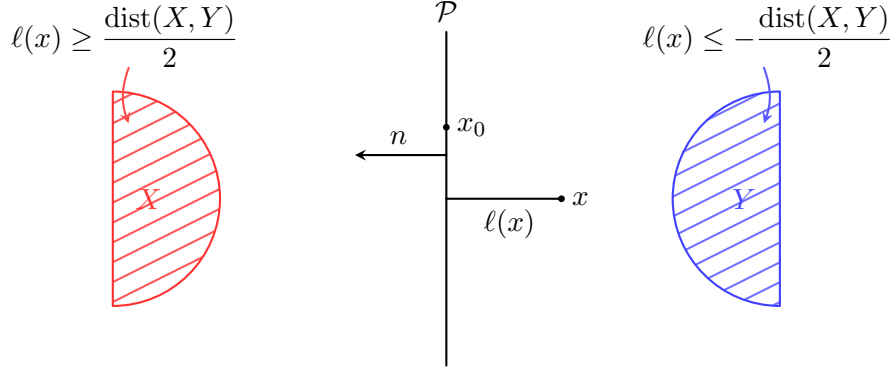
\begin{figure}[h!]
\centering
\begin{tikzpicture}[>=stealth, thick, scale=1.05, line cap=round, line join=round]

  \def\xa{-4.2}
  \def\ya{ 4.2}
  \def\r{ 1.35}

  \coordinate (X)  at (\xa,0);
  \coordinate (Y)  at (\ya,0);
  \coordinate (XR) at (\xa+\r,0);
  \coordinate (YL) at (\ya-\r,0);
  \coordinate (xzero) at (0,0.9);
  \coordinate (xpoint) at (1.45,0);

  \draw (0,-2.1) -- (0,2.1);
  \node[above] at (0,2.1) {$\mathcal P$};

  \fill (xzero) circle (1.2pt);
  \node[right] at (xzero) {$x_0$};

  \draw (0,0) -- (xpoint);
  \fill (xpoint) circle (1.2pt);
  \node[right] at (xpoint) {$x$};
  \node[below] at ($(0,0)!0.55!(xpoint)$) {$\ell(x)$};

  \draw[->] (0,0.55) -- (-1.15,0.55);
  \node[above] at (-0.6,0.55) {$n$};

  \begin{scope}
    \clip (X) ++(0,\r) arc[start angle=90, end angle=-90, radius=\r] -- cycle;
    \foreach \s in {-2.0,-1.75,...,2.0}
      \draw[red!65] ($(X)+(-1.4,\s)$) -- ($(X)+(1.4,\s+1.4)$);
  \end{scope}
  \draw[red!80, thick] (X) ++(0,\r) arc[start angle=90, end angle=-90, radius=\r] -- cycle;
  \node[red!80] at ($(X)+(0.45,0)$) {$X$};

  \begin{scope}
    \clip (Y) ++(0,\r) arc[start angle=90, end angle=270, radius=\r] -- cycle;
    \foreach \s in {-2.0,-1.75,...,2.0}
      \draw[blue!55] ($(Y)+(-1.4,\s)$) -- ($(Y)+(1.4,\s+1.4)$);
  \end{scope}
  \draw[blue!75, thick] (Y) ++(0,\r) arc[start angle=90, end angle=270, radius=\r] -- cycle;
  \node[blue!75] at ($(Y)+(-0.45,0)$) {$Y$};

  \node[align=center] at (-4.05,2.05)
    {$\ell(x)\geq \dfrac{\dist(X,Y)}{2}$};
  \draw[->, red!70] (-4.0,1.65) to[out=-110,in=110] ($(X)+(0.2,0.95)$);

  \node[align=center] at (4.05,2.05)
    {$\ell(x)\leq -\dfrac{\dist(X,Y)}{2}$};
  \draw[->, blue!65] (4.0,1.65) to[out=-70,in=70] ($(Y)+(-0.2,0.95)$);

\end{tikzpicture}
\caption{Schematic illustration of the maximal velocity estimate. The hyperplane $\mathcal P$ is the separating hyperplane passing through $x_0$, and $\ell(x)$ denotes the signed distance from $\mathcal P$ in the chosen normal direction. The red region represents $X$, while the blue region represents $Y$.}
\label{fig:propagation}
\end{figure}

Let $\mathcal{P}$ denote the plane at distance $\dist(X,Y)/2$ from both $X$ and $Y$, $x_0\in \mathcal{P}$, and~$n$ be the unit vector orthogonal to~$\mathcal{P}$ such that~$\ell(x)=n\cdot(x-x_0)\geq \dist(X,Y)/2$ for~$x\in X$ and~$\ell(y)\leq -\dist(X,Y)/2$ for~$y\in Y$.
We also introduce an approximation of~$\ell(x)$, defined by~$\ell_\varepsilon(x)=f_\varepsilon(\ell(x))$ where~$f_\varepsilon(r)=\frac{1}{\varepsilon}f(\varepsilon r)$ and $f\in C^\infty(\bR)$ satisfies $f(r)=r$ on~$[-1,1]$, $0\leq f'\leq 1$ and $f'$ is compactly supported. Note that $f$ is bounded.  Finally, we define two operators:
\begin{equation}\label{def:Geps}
    G_\varepsilon := \Im(\e^{\ell_\varepsilon(x)}\sqrt{-\Delta + 1} \, \e^{-\ell_\varepsilon(x)})
\end{equation}
on $H^1(\bR^3)$ and
\begin{equation}\label{def:G0}
    G_0 :=\mathcal{F} \Im(f(\xi)) \mathcal{F}^{-1},
\end{equation} where $\mathcal{F}$ denotes the Fourier transform and, for $\xi\in \mathbb{R}^3$, $f(\xi):= \sqrt{(\xi+ \i n)^2+1}=\sqrt{\xi^2 + 2 \i n\cdot\xi}$. 
Here we defined the operators in \eqref{def:Geps} and \eqref{def:G0} using quadratic forms, namely for an operator $A$, we consider the operator $\Im(A)$ defined via the corresponding quadratic form
\begin{equation}
    \psi\mapsto\frac{1}{2\i}\langle\psi,\left(A-A^{*}\right)\psi\rangle
\end{equation}
and in our case this is densely defined. 

In the following lemma we summarise some results from \cite[Lemmata 3.3, 3.4 and 3.5]{BFG2025a} which were also used in \cite{BFG2025}: 
\begin{lemma}\label{lem:rappels-BFG2a}
For all $\varepsilon>0$, $G_\varepsilon$ extends to a bounded operator on $L^2(\bR^3)$, with
\begin{equation}\label{eq:bound-Geps}
    \sup_{\varepsilon>0}\|G_\varepsilon\|_{\op}<C
\end{equation}
for some universal constant $C>0$. Moreover, it holds
\begin{equation}\label{eq:bound-G0}
    \|G_0\|_{\op} \leq 1
\end{equation}
and $G_\varepsilon \xrightarrow[\varepsilon\to0]{}G_0$ weakly as bounded operators on $L^2(\bR^3)$.
\end{lemma}

\subsection{Proof of Theorem \ref{thm:maxvel}}
Let us now give the proof of Theorem \ref{thm:maxvel}. 
\begin{proof}[Proof of Theorem \ref{thm:maxvel}]

Let us first consider an initial data $\tilde{\Psi}_{N,0}$ in the form domain  of~$H_N$, that is $H^{1/2}(\mathbb{R}^{3N})$, satisfying $\Id_{X^N}\tilde{\Psi}_{N,0} = \tilde{\Psi}_{N,0}$ and $\|\tilde{\Psi}_{N,0}\|_2=1$. We first prove the result for $\tilde{\Psi}_{N,0}$ and then approximate ${\Psi}_{N,0}$ by a suitable $\tilde{\Psi}_{N,0}$.

\smallskip 

 Let $\tilde{\Psi}_{N,t}=\e^{-\i tH_N} \tilde{\Psi}_{N,0}$ and its reduced one-particle density matrix $\tilde{\gamma}^{(1)}_{N,t}$, we have
\begin{equation}\label{eq:CS-gamma}
    \Tr\big(\Id_{Y}\tilde{\gamma}_{N,t}^{(1)}\big) \leq \|\e^{2\ell_\varepsilon}\Id_{Y}\|_{\op} \, \|\e^{-\ell_\varepsilon(x)} \tilde{\gamma}_{N,t}^{(1)} \e^{-\ell_\varepsilon(x)}\|_{\Tr}.
\end{equation}
For $\varepsilon$ sufficiently small, the operator norm is bounded by $\e^{-\dist(X,Y)}$, and since the operator in the trace norm is nonnegative we have
\begin{equation}\label{eq:CS-gamma2}
    \Tr\big( \Id_{Y}\tilde{\gamma}_{N,t}^{(1)}\big) \leq \e^{-\dist(X,Y)} \, \Tr\big( \e^{-\ell_\varepsilon(x)} \,  \tilde{\gamma}_{N,t}^{(1)} \, \e^{-\ell_\varepsilon(x)}\big)
\end{equation}
for $\varepsilon$ sufficiently small. \\

\noindent {\bf Derivation of an integral expression for $\Tr  \big(\e^{-\ell_\varepsilon(x)} \tilde{\gamma}_{N,t}^{(1)}\e^{-\ell_\varepsilon(x)}\big)$. }
To compute the second factor, we use the fundamental theorem of calculus:
\begin{align*}
    \Tr  \big(\e^{-\ell_\varepsilon(x)} \tilde{\gamma}_{N,t}^{(1)}\e^{-\ell_\varepsilon(x)}\big) 
    &=\int_{\bR^3}\int_{\bR^{3(N-1)}} \e^{-2\ell_\varepsilon(x)}|\tilde{\Psi}_{N,t}(x,Z)|^{2}\,\d{Z}\, \d{x}\\
    &=\int_{\bR^3}\int_{\bR^{3(N-1)}} \e^{-2\ell_\varepsilon(x)}|\tilde{\Psi}_{N,0}(x,Z)|^{2}\d{Z} \, \d{x}\\
    &\quad +\int_{0}^{t}\partial_{\tau} \int_{\bR^3}\int_{\bR^{3(N-1)}} \e^{-2\ell_\varepsilon(x)}|\tilde{\Psi}_{N,\tau}(x,Z)|^{2}\,\d{Z} \, \d{x} \,\d{\tau}.
\end{align*}
The second term can be rewritten as
\begin{align*}
&\int_{0}^{t}\partial_{\tau} \int_{\bR^3}\int_{\bR^{3(N-1)}} \e^{-2\ell_\varepsilon(x)}|\tilde{\Psi}_{N,\tau}(x,Z)|^{2}\,\d{Z} \, \d{x} \, \d\tau\\
    &=2\Re\int_{0}^{t}\left\langle{\e^{-\ell_\varepsilon(x_{1})}\tilde{\Psi}_{N,\tau}},{\e^{-\ell_\varepsilon(x_{1})} \i\left(\sum_{j=1}^{N} \sqrt{1-\Delta_j }+\frac{1}{N-1}\sum_{1\leq j<k\leq N} \!\! w(x_{j}-x_{k})\right) \tilde{\Psi}_{N,\tau}}\right\rangle \, \d\tau\\
    &=2\Re\int_{0}^{t}\left\langle{\e^{-\ell_\varepsilon(x_{1})}\tilde{\Psi}_{N,\tau}},{\i\left(\sum_{j=1}^{N}\e^{-\ell_\varepsilon(x_{1})}\sqrt{1-\Delta_j}\, \e^{\ell_\varepsilon(x_{1})}\right)\e^{-\ell_\varepsilon(x_{1})}\tilde{\Psi}_{N,\tau}}\right\rangle\,\d\tau\\
    &=2\Re\int_{0}^{t}\left\langle{\e^{-\ell_\varepsilon(x_{1})}\tilde{\Psi}_{N,\tau}},{\i\left(\e^{-\ell_\varepsilon(x_{1})}\sqrt{1-\Delta_1}\,\e^{\ell_\varepsilon(x_{1})}\right)\e^{-\ell_\varepsilon(x_{1})}\tilde{\Psi}_{N,\tau}}\right\rangle\,\d\tau\\
    &=2\int_{0}^{t}\left\langle{\e^{-\ell_\varepsilon(x_{1})}\tilde{\Psi}_{N,\tau}},\Im\left(\e^{-\ell_\varepsilon(x_{1})}\sqrt{1-\Delta_1}\,\e^{\ell_\varepsilon(x_{1})}\right)\e^{-\ell_\varepsilon(x_{1})}\tilde{\Psi}_{N,\tau}\right\rangle\,\d\tau\\
    &=2\int_{0}^{t} 
    \Tr  \big(\e^{-\ell_\varepsilon(x)} \tilde{\gamma}_{N,\tau}^{(1)}\e^{-\ell_\varepsilon(x)} G_{\varepsilon} \big) \,\d{\tau}.
\end{align*}
Here we used in the third equality that $\sqrt{-\Delta_j + 1}$ commutes with $\e^{-\ell_\varepsilon(x_{1})}$ for all $j\neq 1$, so for those $j$, we have 
\begin{equation}
    \e^{-\ell_\varepsilon(x_{1})}\sqrt{-\Delta_j + 1}\, \e^{\ell_\varepsilon(x_{1})}=\sqrt{-\Delta_j + 1},
\end{equation}
which is a self-adjoint operator and thus does not contribute due to the $\i$ in the scalar product and the $\Re$ in front of it. 

We obtain
\begin{equation}\label{eq:gamma-tilde-eps}
    \Tr  \big(\e^{-\ell_\varepsilon(x)} \tilde{\gamma}_{N,t}^{(1)}\e^{-\ell_\varepsilon(x)}\big) 
    =
    \Tr  \big(\e^{-\ell_\varepsilon(x)} \tilde{\gamma}_{N,0}^{(1)}\e^{-\ell_\varepsilon(x)}\big) +2\int_{0}^{t} 
    \Tr  \big(\e^{-\ell_\varepsilon(x)} \tilde{\gamma}_{N,\tau}^{(1)}\e^{-\ell_\varepsilon(x)} G_{\varepsilon} \big) \,\d{\tau}.
\end{equation}

\noindent {\bf Grönwall argument for $\Tr  \big(\e^{-\ell_\varepsilon(x)} \tilde{\gamma}_{N,t}^{(1)}\e^{-\ell_\varepsilon(x)}\big)$. }
Using the condition on the support of $\tilde{\Psi}_{N,0}$ and the normalization \eqref{eq:tr-gamma-normalization}, we have
\begin{equation}\label{eq:initalcondest}
   \Tr  \big(\e^{-\ell_\varepsilon(x)} \tilde{\gamma}_{N,0}^{(1)}\e^{-\ell_\varepsilon(x)}\big)= \int_{\bR^3}\int_{\bR^{3(N-1)}} \e^{-2\ell_\varepsilon(x)}|\tilde{\Psi}_{N,0}(x,Z)|^{2}\d Z\d x\\
    \le \e^{-\dist(X,Y)} 
\end{equation}
for every $\varepsilon>0$ and we also get 
\begin{equation}\label{eq:initalcondest-0}
   \Tr  \big(\e^{-\ell(x)} \tilde{\gamma}_{N,0}^{(1)}\e^{-\ell(x)}\big)=   \int_{\bR^3}\int_{\bR^{3(N-1)}} \e^{-2\ell(x)}|\tilde{\Psi}_{N,0}(x,Z)|^{2}\d Z\d x\\
    \le \e^{-\dist(X,Y)} .
\end{equation}
Moreover, since by \eqref{eq:bound-Geps} we have $\sup_{\varepsilon}\|G_{\varepsilon}\|_{\op}\le C$ and since $\e^{-\ell_\varepsilon(x)} \tilde{\gamma}_{N,\tau}^{(1)}\e^{-\ell_\varepsilon(x)}$ is a positive operator, we estimate
\begin{equation}\label{eq:Gepstermest}
    \Tr  \big(\e^{-\ell_\varepsilon(x)} \tilde{\gamma}_{N,\tau}^{(1)}\e^{-\ell_\varepsilon(x)} G_{\varepsilon} \big) 
    \le C  \Tr  \big(\e^{-\ell_\varepsilon(x)} \tilde{\gamma}_{N,\tau}^{(1)}\e^{-\ell_\varepsilon(x)}  \big) .
\end{equation}
Combining \eqref{eq:Gepstermest} and \eqref{eq:gamma-tilde-eps}, we get 
\begin{equation}\label{eq:gamma-tilde-eps2}
    \Tr  \big(\e^{-\ell_\varepsilon(x)} \tilde{\gamma}_{N,t}^{(1)}\e^{-\ell_\varepsilon(x)}\big) 
    \le
    \Tr  \big(\e^{-\ell_\varepsilon(x)} \tilde{\gamma}_{N,0}^{(1)}\e^{-\ell_\varepsilon(x)}\big) +2C\int_{0}^{t} 
    \Tr  \big(\e^{-\ell_\varepsilon(x)} \tilde{\gamma}_{N,\tau}^{(1)}\e^{-\ell_\varepsilon(x)}  \big) \,\d{\tau}.
\end{equation}

By the Gr\"onwall inequality we have 
\begin{equation}\label{eq:gamma-tilde-eps3}
    \Tr  \big(\e^{-\ell_\varepsilon(x)} \tilde{\gamma}_{N,t}^{(1)}\e^{-\ell_\varepsilon(x)}\big) 
    \le
    \Tr  \big(\e^{-\ell_\varepsilon(x)} \tilde{\gamma}_{N,0}^{(1)}\e^{-\ell_\varepsilon(x)}\big) e^{2Ct}\le \e^{2Ct-\dist(X,Y)}, 
\end{equation}
where we used \eqref{eq:initalcondest} in the second inequality. \\

\noindent {\bf A priori bound for $\Tr  \big(\e^{-\ell(x)} \tilde{\gamma}_{N,t}^{(1)}\e^{-\ell(x)}\big) $. }
Now note that 
\begin{align}\label{eq:lepssplit}
    \Tr  \big(\e^{-\ell_\varepsilon(x)} \tilde{\gamma}_{N,t}^{(1)}\e^{-\ell_\varepsilon(x)}\big)
    &= \int_{\ell_\varepsilon(x)\geq0}\int_{\bR^{3(N-1)}}\e^{-2\ell_\varepsilon(x)}|\tilde{\Psi}_{N,t}(x,Z)|^{2}\,\d Z\d x \;\nonumber\\
    &\qquad+\, \int_{\ell_\varepsilon(x)<0} \int_{\bR^{3(N-1)}} \e^{-2\ell_\varepsilon(x)}|\tilde{\Psi}_{N,t}(x,Z)|^{2}\,\d Z\d x .
\end{align}
Using the convergence $\ell_\varepsilon(x)\to \ell(x)$ as $\varepsilon\to0$ for every $x\in\bR$, we apply the dominated convergence theorem for the first term in \eqref{eq:lepssplit}, using that $\|{\Psi}_{N,t}\|_2=1$, and the monotone convergence theorem for the second term in \eqref{eq:lepssplit}. We then find that for every $t\ge0$, we have
\begin{equation}\label{eq:eps0convtr}
    \lim_{\varepsilon\to0} \Tr  \big(\e^{-\ell_\varepsilon(x)} \tilde{\gamma}_{N,t}^{(1)}\e^{-\ell_\varepsilon(x)}\big) = \Tr  \big(\e^{-\ell(x)} \tilde{\gamma}_{N,t}^{(1)}\e^{-\ell(x)}\big).
\end{equation}
Moreover, from \eqref{eq:gamma-tilde-eps3} and \eqref{eq:eps0convtr}, we deduce that
\begin{equation}\label{eq:lpsiest}
      \int_{\mathbb{R}^3} \int_{\bR^{3(N-1)}} \e^{-2\ell(x)}|\tilde{\Psi}_{N,t}(x,Z)|^{2}\,\d Z\d x=\Tr  \big(\e^{-\ell(x)} \tilde{\gamma}_{N,t}^{(1)}\e^{-\ell(x)}\big) 
  \le \e^{2Ct-\dist(X,Y)}.
\end{equation}
The estimate serves as an a priori bound in the following. \\

\noindent {\bf Domination of the integral in  \eqref{eq:gamma-tilde-eps}.
}
We would like to let $\varepsilon\to 0$ in \eqref{eq:gamma-tilde-eps}. 
For all~$\tau\ge0$, applying \eqref{eq:Gepstermest} and \eqref{eq:gamma-tilde-eps3} we get
\begin{equation}
    \big|\Tr  \big(\e^{-\ell_\varepsilon(x)} \tilde{\gamma}_{N,\tau}^{(1)}\e^{-\ell_\varepsilon(x)} G_{\varepsilon} \big)\big|
    \le C\Tr  \big(\e^{-\ell_\varepsilon(x)} \tilde{\gamma}_{N,\tau}^{(1)}\e^{-\ell_\varepsilon(x)}  \big)
    \le C \e^{2C\tau-\dist(X,Y)},
\end{equation}
 Now, since for every $t\ge0$, we have
\begin{equation}
    \e^{2C\tau-\dist(X,Y)}\in L^1_\tau([0,t]),
\end{equation}
we can use dominated convergence in the integral in \eqref{eq:gamma-tilde-eps} when taking the limit $\varepsilon\to0$. Thus, in order to prove
\begin{equation}\label{eq:gamma-tilde-0}
    \Tr  \big(\e^{-\ell(x)} \tilde{\gamma}_{N,t}^{(1)}\e^{-\ell(x)}\big) 
    =
    \Tr  \big(\e^{-\ell(x)} \tilde{\gamma}_{N,0}^{(1)}\e^{-\ell(x)}\big) +2\int_{0}^{t} 
    \Tr  \big(\e^{-\ell(x)} \tilde{\gamma}_{N,\tau}^{(1)}\e^{-\ell(x)} G_{0} \big) \,\d{\tau},
\end{equation}
it remains to prove that for every $\tau\ge0$, we have 
\begin{equation}\label{eq:unterintconvgoal}
    \lim_{\varepsilon\to0}\Tr  \big(\e^{-\ell_\varepsilon(x)} \tilde{\gamma}_{N,\tau}^{(1)}\e^{-\ell_\varepsilon(x)} G_{\varepsilon} \big)
    =
    \Tr  \big(\e^{-\ell(x)} \tilde{\gamma}_{N,\tau}^{(1)}\e^{-\ell(x)} G_{0} \big).
\end{equation}

\noindent {\bf Proof of \eqref{eq:unterintconvgoal}. }
By the cyclicity of the trace and the triangle inequality, we have 
\begin{equation}\label{eq:eps0bigdifference}
\begin{split}
    \big|\Tr&  \big(\e^{-\ell_\varepsilon(x)} \tilde{\gamma}_{N,\tau}^{(1)}\e^{-\ell_\varepsilon(x)} G_{\varepsilon} \big) -
        \Tr  \big(\e^{-\ell(x)} \tilde{\gamma}_{N,\tau}^{(1)}\e^{-\ell(x)} G_{0} \big)\big| \\
    &= \big|\Tr  \big( \tilde{\gamma}_{N,\tau}^{(1)}\e^{-\ell_\varepsilon(x)} G_{\varepsilon} \e^{-\ell_\varepsilon(x)}\big) -
        \Tr  \big( \tilde{\gamma}_{N,\tau}^{(1)}\e^{-\ell(x)} G_{0}\e^{-\ell(x)} \big)\big|\\
    &\le 
    \big|
    \Tr  \big( \tilde{\gamma}_{N,\tau}^{(1)}\e^{-\ell_\varepsilon(x)} G_{\varepsilon} (\e^{-\ell_\varepsilon(x)}-\e^{-\ell(x)})\big)\big|
    +
    \big| \Tr  \big( \tilde{\gamma}_{N,\tau}^{(1)}(\e^{-\ell_\varepsilon(x)}-\e^{-\ell(x)}) G_{\varepsilon} \e^{-\ell(x)}\big)\big|
    \\
    &\qquad +
    \big|
    \Tr  \big( \tilde{\gamma}_{N,\tau}^{(1)}\e^{-\ell(x)} (G_\varepsilon-G_{0})\e^{-\ell(x)} \big)\big|.
\end{split}
\end{equation}

\smallskip

For the first term in \eqref{eq:eps0bigdifference}, since $\tilde{\gamma}_{N,\tau}^{(1)}$ is the projection onto the function $\tilde{\Psi}_{N,t}$, we have
\begin{align*}
    \big|
    \Tr & \big( \tilde{\gamma}_{N,\tau}^{(1)}\e^{-\ell_\varepsilon(x)} G_{\varepsilon} (\e^{-\ell_\varepsilon(x)}-\e^{-\ell(x)})\big)\big|\\
    &=  \int_{\mathbb{R}^3} \int_{\bR^{3(N-1)}} \overline{\tilde{\Psi}_{N,t}(x,Z)} \e^{-\ell_\varepsilon(x)}
    \left( G_\varepsilon\left((\e^{-\ell_\varepsilon(x)}-\e^{-{\ell(x)}})\tilde{\Psi}_{N,t}\right) \right)(x,Z)
    \,\d Z\d x\\
    &\le \sqrt{\int_{\mathbb{R}^3} \int_{\bR^{3(N-1)}} \big|\tilde{\Psi}_{N,t}(x,Z) \e^{-\ell_\varepsilon(x)}\big|^2
    \,\d Z\d x}\\
    &\qquad
    \times\sqrt{\int_{\mathbb{R}^3} \int_{\bR^{3(N-1)}} \big|\left( G_\varepsilon\left((\e^{-\ell_\varepsilon(x)}-\e^{-\ell(x)})\tilde{\Psi}_{N,t}\right)\right)(x,Z)\big|^2
    \,\d Z\d x}\\
    &\le \e^{C\tau-\frac{1}{2}\dist(X,Y)}\|G_{\varepsilon}\|_{\op}\|(\e^{-\ell_\varepsilon(x)}-\e^{-\ell(x)})\tilde{\Psi}_{N,\tau}\|_2.
\end{align*}
Here we used \eqref{eq:gamma-tilde-eps3}. Note that
\begin{equation}\label{eq:lepslpsitozero}
    \lim_{\varepsilon\to0}
    \|(\e^{-\ell_\varepsilon(x)}-\e^{-\ell(x)})\tilde{\Psi}_{N,\tau}\|_2=0
\end{equation}
by the dominated convergence theorem with dominating function
\begin{equation}
    2(1+\e^{-\ell(x)})^2|\tilde{\Psi}_{N,\tau}(x,Z)|^2.
\end{equation}
This shows that the first term in \eqref{eq:eps0bigdifference} goes to zero as $\varepsilon\to0$.

\smallskip

For the second term in \eqref{eq:eps0bigdifference}, we perform a similar computation:
\begin{align}
\begin{split}
    \big| \Tr & \big( \tilde{\gamma}_{N,\tau}^{(1)}(\e^{-\ell_\varepsilon(x)}-\e^{-\ell(x)}) G_{\varepsilon} \e^{-\ell(x)}\big)\big|\\
    &=  \Big | \int_{\mathbb{R}^3} \int_{\bR^{3(N-1)}} \overline{\tilde{\Psi}_{N,t}(x,Z)} (\e^{-\ell_\varepsilon(x)}-\e^{-\ell(x)})
    \left( G_\varepsilon(\e^{-\ell(x)}\tilde{\Psi}_{N,t} )\right)(x,Z)
    \,\d Z\d x \Big |\\
    &\le \|(\e^{-\ell_\varepsilon(x)}-\e^{-\ell(x)})\tilde{\Psi}_{N,\tau}\|_2\|G_{\varepsilon}\|_{\op}\|\e^{-\ell(x)}\tilde{\Psi}_{N,t}\|_2.
\end{split}
\end{align}
By \eqref{eq:lepslpsitozero} and \eqref{eq:bound-Geps}, we find that the second term in \eqref{eq:eps0bigdifference} also goes to zero as $\varepsilon\to0$.

\smallskip

For the third term in \eqref{eq:eps0bigdifference}, note that 
\begin{align}
\begin{split}
    \lim_{\varepsilon\to0}&\big|
    \Tr  \big( \tilde{\gamma}_{N,\tau}^{(1)}\e^{-\ell(x)} (G_\varepsilon-G_{0})\e^{-\ell(x)} \big)\big|\\
    &=  \lim_{\varepsilon\to0} \left|\int_{\mathbb{R}^3} \int_{\bR^{3(N-1)}} \overline{\tilde{\Psi}_{N,t}(x,Z)} \e^{-\ell(x)}
    \left( \left(G_\varepsilon-G_0\right)\left(\e^{-\ell(x)}\tilde{\Psi}_{N,t} \right)\right)(x,Z)
    \,\d Z\d x \right|\\
    &=0.
\end{split}
\end{align}
Here we used the weak convergence of $G_\epsilon$ to $G_0$ that was proved in Lemma \ref{lem:rappels-BFG2a} on the one-particle level, but which extends to these operators considered as operators acting only on the first variable, and we also used that 
$\e^{-\ell(x_1)}\tilde{\Psi}_{N,t}\in L^2(\bR^{3N}) $ by \eqref{eq:lpsiest}. 

\smallskip

To sum up, we have now proved \eqref{eq:unterintconvgoal}, which was sufficient to prove \eqref{eq:gamma-tilde-0}. \\

\noindent {\bf Grönwall argument for $\Tr  \big(\e^{-\ell(x)} \tilde{\gamma}_{N,t}^{(1)}\e^{-\ell(x)}\big)$. }
Next, recall from Lemma \ref{lem:rappels-BFG2a} that 
\begin{equation}
    \|G_{0}\|_{\op}\le1.
\end{equation}
Thus, for every $\tau\ge0$, by positivity of $\e^{-\ell(x)} \tilde{\gamma}_{N,\tau}^{(1)}\e^{-\ell(x)}$, we have
\begin{equation}
     \Tr  \big(\e^{-\ell(x)} \tilde{\gamma}_{N,\tau}^{(1)}\e^{-\ell(x)} G_{0} \big)\le  \Tr  \big(\e^{-\ell(x)} \tilde{\gamma}_{N,\tau}^{(1)}\e^{-\ell(x)}  \big).
\end{equation}
Combining this with \eqref{eq:gamma-tilde-0}, we find that 
\begin{equation}\label{eq:withoutG0beforeGronwall}
    \Tr  \big(\e^{-\ell(x)} \tilde{\gamma}_{N,t}^{(1)}\e^{-\ell(x)}\big) 
    \le
    \Tr  \big(\e^{-\ell(x)} \tilde{\gamma}_{N,0}^{(1)}\e^{-\ell(x)}\big) +2\int_{0}^{t} 
    \Tr  \big(\e^{-\ell(x)} \tilde{\gamma}_{N,\tau}^{(1)}\e^{-\ell(x)}  \big) \,\d{\tau}.
\end{equation}
We now apply Grönwall's inequality to \eqref{eq:withoutG0beforeGronwall}, and then use \eqref{eq:initalcondest-0} to obtain
\begin{equation}\label{eq:gamma-tilde-0-est-Gro}
    \Tr  \big(\e^{-\ell(x)} \tilde{\gamma}_{N,t}^{(1)}\e^{-\ell(x)}\big) 
    \le
    \Tr  \big(\e^{-\ell(x)} \tilde{\gamma}_{N,0}^{(1)}\e^{-\ell(x)}\big)\e^{2t}
    \le \e^{2t-\dist(X,Y)}.
\end{equation}

\smallskip

\noindent {\bf Conclusion for $\tilde\Psi_{N,0}\in H^{1/2}(\bR^{3N})$. }
Now recall from \eqref{eq:CS-gamma2} that for $\varepsilon$ sufficiently small, we have
\begin{equation}\label{eq:CS-gamma3}
    \Tr\big( \Id_{Y}\tilde{\gamma}_{N,t}^{(1)}\big) \leq \e^{-\dist(X,Y)} \, \Tr\big( \e^{-\ell_\varepsilon(x)} \tilde{\gamma}_{N,t}^{(1)} \e^{-\ell_\varepsilon(x)}\big),
\end{equation}
and by \eqref{eq:eps0convtr}, we have
\begin{equation}\label{eq:eps0convtr1}
    \lim_{\varepsilon\to0} \Tr  \big(\e^{-\ell_\varepsilon(x)} \tilde{\gamma}_{N,t}^{(1)}\e^{-\ell_\varepsilon(x)}\big) = \Tr  \big(\e^{-\ell(x)} \tilde{\gamma}_{N,t}^{(1)}\e^{-\ell(x)}\big).
\end{equation}
Thus, combining \eqref{eq:gamma-tilde-0-est-Gro}, \eqref{eq:CS-gamma3} and \eqref{eq:eps0convtr1}, we obtain 
\begin{equation}
    \Tr\big( \Id_{Y}\tilde{\gamma}_{N,t}^{(1)}\big)\le \e^{2\left(t-\dist(X,Y)\right)}.
\end{equation}
This concludes the proof of Theorem \ref{thm:maxvel} for $\tilde\Psi_{N,0}\in H^{1/2}(\bR^{3N})$. \\

\noindent {\bf Approximation argument. }
It remains to deduce the general result using  an approximation argument. To this end, let us now consider $\Psi_{N,0}\in L^2_{\mathrm{sym}}(\bR^{3N})$ with $\Id_{X^N}\Psi_{N,0} = \Psi_{N,0}$ and $\|\Psi_{N,0}\|_2=1$. For any $\delta>0$, there exists $\tilde{\Psi}_{N,0}\in H^{1/2}_{\mathrm{sym}}(\bR^{3N})$ with $\Id_{X^N}\tilde{\Psi}_{N,0} = \tilde{\Psi}_{N,0}$, $\|\Psi_{N,0} - \tilde{\Psi}_{N,0}\|_{2} < \delta$ and $\|\tilde\Psi_{N,0}\|_2=1$. 

Thus, denoting by $\tilde\gamma_{N,t}^{(1)}$ and $\gamma_{N,t}^{(1)}$ the reduced density matrices of $\tilde{\Psi}_{N,t}$ and $\Psi_{N,t}$ respectively, it holds
\begin{align*}
    \|\Id_Y  \gamma_{N,t}^{(1)}\|_{\mathrm{Tr}} 
    &\leq  \|\Id_Y  (\tilde\gamma_{N,t}^{(1)}-\gamma_{N,t}^{(1)})\|_{\mathrm{Tr}}+ \|\Id_Y  \tilde{\gamma}_{N,t}^{(1)}\|_{\mathrm{Tr}}  \\
    &\le \|\tilde\gamma_{N,t}^{(1)}-\gamma_{N,t}^{(1)}\|_{\mathrm{Tr}}+\e^{2\left(t-\dist(X,Y)\right)}\\
    & \leq (\|\Psi_{N,0}\|_{2}+\|\tilde{\Psi}_{N,0}\|_{2}) \|\Psi_{N,0}-\tilde{\Psi}_{N,0}\|_{2} + \e^{2(t-\dist(X,Y))} \\
    & \leq 2 \delta + \e^{2(t-\dist(X,Y))} .
\end{align*}
As this is true for $\delta$ arbitrarily small, we get the result.
\end{proof}

\section{Many-Body Minimal Velocity Estimate}\label{sec:manybody-min-vel}

In this section we prove a many-body minimal velocity estimate as stated in Theorem \ref{thm:intro-many-body-minvel}.

We recall the velocity observable, defined on $\bR^3$,
\[
    \Theta
    :=
    \frac{-\i\nabla}{\sqrt{-\Delta+1}},
    \qquad
    \Theta_j
    :=
    \frac{-\i\nabla_j}{\sqrt{-\Delta_j+1}},
\]
and define the one-particle observable
\begin{equation}\label{def:many-body-pc-observable}
    A_j(t):=x_j-t\Theta_j.
\end{equation}
We also define the corresponding one-particle moment by
\begin{equation}\label{def:many-body-pc-moment}
    \mathcal{M}_N(t)
    :=
    \langle \Psi_{N,t}, |A_1(t)|^2\Psi_{N,t}\rangle .
\end{equation}
Equivalently,
\begin{equation}\label{eq:pc-moment-asymptotic-velocity}
    \frac{1}{t^2}\mathcal{M}_N(t)
    =
    \left\|
        \left(\frac{x_1}{t}-\Theta_1\right)\Psi_{N,t}
    \right\|_2^2 .
\end{equation}

\subsection{Proof of Theorem \ref{thm:intro-many-body-minvel}\eqref{thm:minimal-L2bound}}
\label{subsec:MomentBound}

In this section, we prove the moment estimate Theorem \ref{thm:intro-many-body-minvel}\eqref{thm:minimal-L2bound}
\begin{equation}\label{eq:asymptotic-velocity-bound}
    \left\|
        \left(\frac{x_1}{t}-\Theta_1\right)
        \Psi_{N,t}
    \right\|_2^2
 \leq
    C\left(\frac{1}{t^2}+
        \frac{1}{t} 
        +
        \frac{t^2}{N}
    \right).
\end{equation}
This will also be the main input for the proof of  Theorem \ref{thm:intro-many-body-minvel}\eqref{thm:minimal-trace-bound} in the following subsection. Our proof of Theorem \ref{thm:intro-many-body-minvel}\eqref{thm:minimal-L2bound} uses the mean-field estimate \eqref{eq:alphaNest} from Section~\ref{sec:mean-field}
\[
    \alpha_N(t)\leq \frac{C}{N}.
\]
In Appendix \ref{sec:one-body-minvel-direct}, we rewrote the proof of the corresponding statement of Theorem \ref{thm:intro-many-body-minvel}\eqref{thm:minimal-L2bound} on the one-body level for the boson star equation \eqref{eq:BSE}. At this point, we would like to mention that this proof idea differs significantly from the arguments used in   \cite[Proposition 8.7]{BFG2025}.

\begin{proof}[Proof of Theorem \ref{thm:intro-many-body-minvel}\eqref{thm:minimal-L2bound}]
Since $A_1(t)=x_1-t\Theta_1$ solves the free Heisenberg equation
\[
    \partial_t A_1(t)
    +
    \i\big[\sqrt{-\Delta_1+1},A_1(t)\big]
    =
    0,
\]
only the interaction part of the Hamiltonian contributes to the time derivative of $\mathcal{M}_N(t)$, which was defined in \eqref{def:many-body-pc-moment}. Using the symmetry of $\Psi_{N,t}$ and the normalization $(N-1)^{-1}$ in the interaction (see Appendix \ref{sec:one-body-minvel-direct}) we obtain
\begin{equation}\label{eq:beta-derivative}
    \frac{\d}{\d t}\mathcal{M}_N(t)
    =
    2t\,\Re
    \left\langle
        A_1(t)\Psi_{N,t},
        \i[w_{12},\Theta_1]\Psi_{N,t}
    \right\rangle,
\end{equation}
where
\[
    w_{12}:=w(x_1-x_2).
\]

We insert the decomposition $1=\p_2+\q_2$ with $\p_2, \q_2$ defined in \eqref{eq:pqdef} 
and we get 
\begin{equation}\label{eq:pq-decomposition-comm}
    [w_{12},\Theta_1]\Psi_{N,t}
    =
    [w_{12},\Theta_1]\p_2\Psi_{N,t}
    +
    [w_{12},\Theta_1]\q_2\Psi_{N,t}.
\end{equation}
Since $\Theta_1$ acts on the first particle, while $\p_2$ and $\q_2$ act on the second particle, we have
\[
    [\Theta_1,\p_2]=[\Theta_1,\q_2]=0.
\]
We now estimate these three terms in \eqref{eq:pq-decomposition-comm} separately.

For the first term, since $\p_2$ commutes with $\Theta_1$,
\[
    [w_{12},\Theta_1]\p_2
    =
    [w_{12}\p_2,\Theta_1],
\]
and
\begin{equation}\label{eq:bound-commutator-w12p2-Theta}
    \|[w_{12}\p_2,\Theta_1]\|_{\mathrm{op}}
\leq 2\|w_{12}\p_2\|_{\mathrm{op}}
= 2 \|\p_2 w_{12}^2 \p_2\|^{1/2}_{\mathrm{op}}
\leq 2 \|u_t\|_{\infty} \|w\|_2
\end{equation}
since, by the same method as in \eqref{eq:estimate-p-Vt2-p},
\[
\|\p_2 w_{12}^2 \p_2\|_{\mathrm{op}}
= \|\p_2 (w^2 * |u_t|^2) \p_2\|_{\mathrm{op}}
\leq \|w^2 * |u_t|^2\|_{\infty}
\leq \|w\|_2^2 \|u_t\|^2_\infty.
\]

Consequently,
\begin{equation}\label{eq:mixed-contribution}
    \left|
    \left\langle
        A_1(t)\Psi_{N,t},
        \i[w_{12},\Theta_1]\p_2\Psi_{N,t}
    \right\rangle
    \right|
    \leq
    C\|w\|_2\langle t\rangle^{-3/2}
    \sqrt{\mathcal{M}_N(t)}.
\end{equation}
Finally, we estimate the last term in \eqref{eq:pq-decomposition-comm}. Since $w\in L^\infty(\bR^3)$ and $\|\Theta_1\|_{\op}\leq1$,
\[
    \|[w_{12},\Theta_1]\|_{\op}
    \leq
    2\|w\|_\infty.
\]
Thus, by the Cauchy--Schwarz inequality,
\begin{align}
    \left|
    \left\langle
        A_1(t)\Psi_{N,t},
        \i[w_{12},\Theta_1]\q_2\Psi_{N,t}
    \right\rangle
    \right|
    &\leq
    2\|w\|_\infty
    \|A_1(t)\Psi_{N,t}\|_2
    \|\q_2\Psi_{N,t}\|_2
    \nonumber\\
    &=
    2\|w\|_\infty
    \sqrt{\mathcal{M}_N(t)}
    \sqrt{\alpha_N(t)},
    \label{eq:fluctuation-contribution-raw}
\end{align}
where $\alpha_N(t)$ is defined in \eqref{eq:alpha_symmetry}. 
Using the uniform mean-field bound, see Proposition \ref{lemma: bound alpha},
\[
    \alpha_N(t)\leq \frac{C}{N},
\]
we get
\begin{equation}\label{eq:fluctuation-contribution}
    \left|
    \left\langle
        A_1(t)\Psi_{N,t},
        \i[w_{12},\Theta_1]\q_2\Psi_{N,t}
    \right\rangle
    \right|
    \leq
    \frac{C\|w\|_\infty}{\sqrt N}
    \sqrt{\mathcal{M}_N(t)}.
\end{equation}

Combining \eqref{eq:beta-derivative}-\eqref{eq:bound-commutator-w12p2-Theta}, \eqref{eq:mixed-contribution}, and \eqref{eq:fluctuation-contribution}, we obtain
\[
    \frac{\d}{\d t}\mathcal{M}_N(t)
    \leq
    C\|w\|_2t\langle t\rangle^{-3/2}
    \sqrt{\mathcal{M}_N(t)}
    +
    C\|w\|_\infty\frac{t}{\sqrt N}
    \sqrt{\mathcal{M}_N(t)}.
\]
Thus, for all $t>0$ and $\varepsilon>0$,
\[
    \frac{\d}{\d t}\sqrt{\mathcal{M}_N(t)+\varepsilon}
    \leq
    C\|w\|_2t\langle t\rangle^{-3/2}
    +
    C\|w\|_\infty\frac{t}{\sqrt N}.
\]
Integrating from $0$ to $t$ and then letting $\varepsilon\to0$ gives
\[
    \sqrt{\mathcal{M}_N(t)}
    \leq
    \sqrt{\mathcal{M}_N(0)}
    +
    C\|w\|_2\int_0^t \tau\langle\tau\rangle^{-3/2}\,\d\tau
    +
    \frac{C\|w\|_\infty}{\sqrt N}
    \int_0^t \tau\,\d\tau .
\]
The first integral is bounded by $C\langle t\rangle^{1/2}$, and the second is $t^2/2$. Moreover, since $\Psi_{N,0}=u_0^{\otimes N}$ and $A_1(0)=x_1$,
\[
    \mathcal{M}_N(0)=\|xu_0\|_2^2.
\]
Hence
\begin{equation}
    \sqrt{\mathcal{M}_N(t)}
    \leq
    C\left(\|xu_0\|_2
        +\|w\|_2 t^{1/2}
        +
        \frac{\|w\|_\infty t^2}{\sqrt N}
    \right),
\end{equation}
and taking the square gives
\begin{equation}
    \mathcal{M}_N(t)
    \leq
    C\left(\|xu_0\|_2^2
        +\|w\|_2^2 t
        +
        \frac{\|w\|_\infty^2 t^4}{N}
    \right).
 \end{equation}
 Finally, we have
\begin{equation}
    \left\|
        \left(\frac{x_1}{t}-\Theta_1\right)
        \Psi_{N,t}
    \right\|_2^2
    =
    \frac{1}{t^2}\mathcal{M}_N(t)\leq
    C\left(\frac{1}{t^2}+
    \frac{1}{t} 
        +
        \frac{ t^2}{N}
    \right),
\end{equation}
where $C>0$ depends on $w$ and $u_0$. This gives \eqref{eq:asymptotic-velocity-bound} and completes the proof of Theorem \ref{thm:intro-many-body-minvel}\eqref{thm:minimal-L2bound}.
\end{proof}

\subsection{Proof of Theorem \ref{thm:intro-many-body-minvel}\eqref{thm:minimal-trace-bound}}\label{subsec:min-vel-est}
We now deduce the many-body minimal velocity estimate stated in Theorem \ref{thm:intro-many-body-minvel}\eqref{thm:minimal-trace-bound} from the moment bound from Theorem \ref{thm:intro-many-body-minvel}\eqref{thm:minimal-L2bound}.

\begin{proof}[Proof of Theorem \ref{thm:intro-many-body-minvel}\eqref{thm:minimal-trace-bound}]
By definition of $\gamma^{(1)}_{N,t}$ we have
\[
    \Tr\left[
        g(\Theta^2)
        f^2\left(\frac{|x|^2}{t^2}\right)
        g(\Theta^2)
        \gamma^{(1)}_{N,t}
    \right]
    =
    \left\|
        f\left(\frac{|x_1|^2}{t^2}\right)
        g(\Theta_1^2)
        \Psi_{N,t}
    \right\|_2^2.
\]
The functions $f$ and $g$ have disjoint supports. Hence this observable localizes to the region where the asymptotic position $x/t$ and the velocity observable $\Theta$ are incompatible.

Using the Helffer-Sj\"ostrand functional calculus, it is proven in \cite[Proposition 8.8]{BFG2025} (see, more precisely, the last equation in the proof of Proposition 8.8 in \cite{BFG2025}) that,
\[
    \left\|
        \left(g\left(\frac{|x|^2}{t^2}\right)
        -g(\Theta^2)\right)
        \varphi
    \right\|_2
    \le
    C_g t^{-1} \|\varphi\|_2 + C_g\left\|\left(\Theta-\frac{x}{t}\right)\varphi\right\|_2 ,
\]
for all $\varphi\in L^2(\mathbb{R}^3)$ such that $|x|\varphi\in L^2(\mathbb{R}^3)$. Since $f$ and $g$ have disjoint supports $f(|x|^2/t^2)g(|x|^2/t^2)=0$. Hence, by the preceding functional calculus estimate and since $\|\Psi_{N,t}\|_2=1$, we obtain
\[
\begin{aligned}
    \left\|
        f\left(\frac{|x_1|^2}{t^2}\right)
        g(\Theta_1^2)
        \Psi_{N,t}
    \right\|_2^2
    &=
    \left\|
        f\left(\frac{|x_1|^2}{t^2}\right)\left(g\left(\frac{|x_1|^2}{t^2}\right)-
        g(\Theta_1^2)\right)
        \Psi_{N,t}
    \right\|_2^2\\
    &\le  
    C_{f,g}t^{-2}
    +
    C_{f,g} \left\|
        \left(\frac{x_1}{t}-\Theta_1\right)\Psi_{N,t}
    \right\|_2^2.
\end{aligned}
\]
Applying \eqref{eq:asymptotic-velocity-bound},
this establishes Theorem \ref{thm:intro-many-body-minvel}\eqref{thm:minimal-trace-bound}.
\end{proof}

\subsection{Proof of Theorem \ref{thm:intro-many-body-minvel}\eqref{thm:minimal-nu-bound}}\label{subsec:nuN}

In this section, we provide a weighted estimate for the number of particles $\nu_N(t)$ outside the condensate $u_t$. Note that the definition of $\nu_N(t)$ given in the statement of Theorem \ref{thm:intro-many-body-minvel}\eqref{thm:minimal-nu-bound} is equivalent to 
\begin{equation}\label{def:nuN}
    \nu_N(t)
    :=
    \|x_1\q_1\Psi_{N,t}\|_2^2,
\end{equation}
where $\q_1$ was defined in \eqref{eq:pqdef}.

\begin{remark}\phantom{}
\begin{enumerate}

\item 
The normalized quantity $t^{-2}\nu_N(t)$ measures the second moment of the asymptotic position variable $x_1/t$ for the part of the wave function outside the condensate. In a scattering regime, $x_1/t$ may be interpreted as an asymptotic velocity, or group-velocity-type variable. 
The estimate
\begin{equation}
    \nu_N(t)\le Ct^2
\end{equation}
from  Theorem \ref{thm:intro-many-body-minvel}\eqref{thm:minimal-nu-bound}
shows that the weighted bad-particle component spreads at most linearly in time. Equivalently, the asymptotic-velocity-type quantity~$t^{-2}\nu_N(t)$ remains uniformly controlled.

\item Also note that Theorem \ref{thm:intro-many-body-minvel}\eqref{thm:minimal-L2bound} does \emph{not} imply \ref{thm:intro-many-body-minvel}\eqref{thm:minimal-nu-bound} because the estimate 
\begin{equation}
    \left\|
        \left(\frac{x_1}{t}-\Theta_1\right)
        \Psi_{N,t}
    \right\|_2^2
\leq
    C\left(\frac{1}{t^2}
        +\frac{1}{t} 
        +
        \frac{t^2}{N}
    \right)
\end{equation}
contains a $ t^2/N$ term, which would yield a $ t^4/N$ term after multiplying by $t^2$.
\end{enumerate}
\end{remark}

\begin{proof}[Proof of Theorem \ref{thm:intro-many-body-minvel}\eqref{thm:minimal-nu-bound}]
We use the alternative definition of $\nu_N(t)$ given in \eqref{def:nuN}.
Since~$u_t$ solves the Hartree equation
\[
    \i\partial_t u_t=h^{\MF}(t)u_t,
    \qquad
    h^{\MF}(t)=\sqrt{-\Delta+1}+ (w*|u_t|^2),
\]
the time-dependent projector $\q_1$ satisfies
\[
    \partial_t\q_1=-\i[h^{\MF}_1(t),\q_1].
\]
Using also $\i\partial_t\Psi_{N,t}=H_N\Psi_{N,t}$, we obtain
\begin{align}
    \frac{\d}{\d t}\nu_N(t)
    &=
    \langle \Psi_{N,t},
        \i[H_N,\q_1x_1^2\q_1]\Psi_{N,t}
    \rangle
    +
    \langle \Psi_{N,t},
        \partial_t(\q_1x_1^2\q_1)\Psi_{N,t}
    \rangle
    \nonumber\\
    &=
    \langle \Psi_{N,t},
        \i[H_N-h^{\MF}_1,\q_1]x_1^2\q_1\Psi_{N,t}
    \rangle
    +
    \langle \Psi_{N,t},
        \i\q_1x_1^2[H_N-h^{\MF}_1,\q_1]\Psi_{N,t}
    \rangle
    \nonumber\\
    &\quad
    +
    \langle \q_1\Psi_{N,t},
        \i[H_N,x_1^2]\q_1\Psi_{N,t}
    \rangle .
    \label{eq:nuN-derivative-decomposition}
\end{align}

We first estimate the last term. Since the interaction potential  commutes with $x_1^2$, by passing to Fourier space, we have
\[
    \i[H_N,x_1^2]
    =
    \i[\sqrt{-\Delta_1+1},x_1^2]
    =
    x_1\cdot \Theta_{1}+\Theta_{1}\cdot x_1,
\]
where we recall the definition
\[
    \Theta_{1}:=\frac{-\i\nabla_{1}}{\sqrt{-\Delta_1 + 1}}.
\]
In particular $\|\Theta_{1}\|_{\op}\le 1$ and therefore
\begin{align}
    \left|
    \langle \q_1\Psi_{N,t},
        \i[H_N,x_1^2]\q_1\Psi_{N,t}
    \rangle
    \right|
    &\le
    2\|x_1\q_1\Psi_{N,t}\|_2
      \|\Theta_{1}\q_1\Psi_{N,t}\|_2
    \nonumber\\
    &\le
    2\sqrt{\nu_N(t)}\,\|\q_1\Psi_{N,t}\|_2
    =
    2\sqrt{\nu_N(t)}\,\sqrt{\alpha_N(t)} ,
    \label{eq:nuN-kinetic-term}
\end{align}
with $\alpha_N(t)$ defined in \eqref{eq:alpha_symmetry}. 

Next, we estimate the terms containing commutators with $\q_1$. As computed in \eqref{comp: [H - h, q1]} it holds 
\[
    [H_N-h^{\MF}_1,\q_1]
    =
    [V_t(x_1,x_2),\q_1],
\]
with $V_t(x_1,x_2)$ defined in \eqref{def:Vt}. 
Using that $[x_1^2,V_t(x_1,x_2)]=0$, 
 the first two terms in~\eqref{eq:nuN-derivative-decomposition} become
\begin{multline}
    \langle \Psi_{N,t},
        \i[H_N-h^{\MF}_1,\q_1]x_1^2\q_1\Psi_{N,t}
    \rangle
    +
    \langle \Psi_{N,t},
        \i\q_1x_1^2[H_N-h^{\MF}_1,\q_1]\Psi_{N,t}
    \rangle
    \\
    =-
    2\,
    \Im
    \langle
        \Psi_{N,t},
        V_t(x_1,x_2)\q_1x_1^2\q_1\Psi_{N,t}
    \rangle .
    \label{eq:nuN-interaction-term}
\end{multline}
We now expand
\[
    \q_1x_1^2\q_1
    =
    x_1^2-x_1^2\p_1-\p_1x_1^2+\p_1x_1^2\p_1.
\]
Since $V_t$ is real-valued, we have
\[
    \Im\langle \Psi_{N,t},V_t(x_1,x_2)x_1^2\Psi_{N,t}\rangle=0.
\]
Thus, inserting $1=\p_1 + \q_1$, it holds 
\begin{align*}
    &\Im \langle \Psi_{N,t}, \, V_t(x_1, x_2) \q_1 x_1^2 \q_1 \Psi_{N,t}\rangle\\
    &=   - \Im \langle \Psi_{N,t}, \,  (\p_1 +\q_1) V_t(x_1, x_2) x_1^2 \p_1 \Psi_{N,t}\rangle
    - \Im \langle \Psi_{N,t}, \, V_t(x_1, x_2)  \p_1 x_1^2 (\p_1 +\q_1) \Psi_{N,t} \rangle \\
    &\qquad+\Im \langle \Psi_{N,t}, \, V_t(x_1, x_2)  \p_1 x_1^2 \p_1 \Psi_{N,t} \rangle \\
    &= - \Im \langle \Psi_{N,t}, \,  \q_1 V_t(x_1, x_2) x_1^2 \p_1 \Psi_{N,t}\rangle- \Im \langle \Psi_{N,t}, \,   V_t(x_1, x_2) \p_1 x_1^2 \q_1 \Psi_{N,t}\rangle,
\end{align*}
where we used that $ \Im \langle \Psi_{N,t}, \,   \p_1 V_t(x_1, x_2)  x_1^2 \p_1 \Psi_{N,t}\rangle=0$. 
By Cauchy-Schwarz inequality, we obtain
\begin{align}
    \left|
    \Im
    \langle
        \Psi_{N,t},
        V_t(x_1,x_2)\q_1x_1^2\q_1\Psi_{N,t}
    \rangle
    \right|
    &\le
    \sqrt{\nu_N(t)}
    \Big(
        \|V_t(x_1,x_2)x_1\p_1\Psi_{N,t}\|_2
        \nonumber\\
        &\hspace{4cm}
        +
        \|x_1\p_1V_t(x_1,x_2)\Psi_{N,t}\|_2
    \Big).
    \label{eq:nuN-basic-interaction-bound}
\end{align}

As we prove in Lemma \ref{lem:interaction-weighted-q} below, there exists a constant $C>0$, depending only on $w$ and $u_0$, such that
\begin{equation}\label{eq:Vtx1est}
    \|V_t(x_1,x_2)x_1\p_1\Psi_{N,t}\|_2
    +
    \|x_1\p_1V_t(x_1,x_2)\Psi_{N,t}\|_2
    \le C
\end{equation}
for all $t\ge0$ and all $N$. 

Therefore, from \eqref{eq:nuN-interaction-term},
\eqref{eq:nuN-basic-interaction-bound} and \eqref{eq:Vtx1est}, we get
\begin{equation}
    \left|
    \langle \Psi_{N,t},
        \i[H_N-h^{\MF}_1,\q_1]x_1^2\q_1\Psi_{N,t}
    \rangle
    +
    \langle \Psi_{N,t},
        \i\q_1x_1^2[H_N-h^{\MF}_1,\q_1]\Psi_{N,t}
    \rangle
    \right|
    \le
    C\sqrt{\nu_N(t)} .
    \label{eq:nuN-interaction-final}
\end{equation}

Combining \eqref{eq:nuN-derivative-decomposition},
\eqref{eq:nuN-kinetic-term}, and \eqref{eq:nuN-interaction-final}, we find
\[
    \frac{\d}{\d t}\nu_N(t)
    \le
    2\sqrt{\nu_N(t)}\,\sqrt{\alpha_N(t)}
    +
    C\sqrt{\nu_N(t)}.
\]
By the uniform-in-time mean-field estimate, see Proposition \ref{lemma: bound alpha}, it holds
\[
    \alpha_N(t)\le \frac{C}{N},
\]
and hence
\[
    \frac{\d}{\d t}\nu_N(t)
    \le
    C\sqrt{\nu_N(t)}
    \left(N^{-1/2}+1\right)\le C\sqrt{\nu_N(t)} .
\]
Thus, for every $t\ge0$ with $\nu_N(t)>0$,
\[
    \frac{\d}{\d t}\sqrt{\nu_N(t)}
    =
    \frac{1}{2\sqrt{\nu_N(t)}}\, \frac{\d}{\d t}\nu_N(t)
    \le
   C.
\]
Since the initial datum is factorized, i.e., $\q_1\Psi_{N,0}=0$, we have $\nu_N(0)=0$ and therefore  Gr\"onwall's inequality implies
\begin{equation}
    \sqrt{\nu_N(t)}\le Ct.
\end{equation}
Taking the square proves \eqref{eq:bound-nuN}.
\end{proof}

In order to complete the proof of Theorem \ref{thm:intro-many-body-minvel}\eqref{thm:minimal-nu-bound}, it remains to show \eqref{eq:Vtx1est}. This is done in the following lemma.

\begin{lemma}[Interaction estimate]\label{lem:interaction-weighted-q}
Assuming that Hypothesis \ref{hyp} and the bounds in~\eqref{eq:weighted1bodydecay} hold, there exists $C>0$, independent of $N$ and $t$, such that
\[
    \|V_t(x_1,x_2)x_1\p_1\Psi_{N,t}\|_2
    +
    \|x_1\p_1V_t(x_1,x_2)\Psi_{N,t}\|_2
    \le C.
\]
\end{lemma}

\begin{proof}

For the first term, we have
\[
\begin{aligned}
    \|V_t(x_1,x_2)x_1\p_1\Psi_{N,t}\|_2^2
    &=
    \langle\Psi_{N,t},
        \p_1 x_1 V_t(x_1,x_2)^2 x_1\p_1
        \Psi_{N,t}
    \rangle                                      \\
    &\le
    \|\p_1 x_1 V_t(x_1,x_2)^2 x_1\p_1\|_{\op}.
\end{aligned}
\]
Moreover, we set
\[
    \p_1 x_1 V_t(x_1,x_2)^2 x_1\p_1
    =
    h_t(x_2)\p_1,
\]
where
\[
    h_t(x_2)
    =
    \int_{\mathbb R^3}
        |x_1u_t(x_1)|^2
        |V_t(x_1,x_2)|^2
    \,\dd x_1 .
\]
Therefore we have
\[
    \|h_t\|_\infty
    \le
    \|xu_t\|_\infty^2
    \sup_{x_2}\|V_t(\cdot,x_2)\|_2^2
    \le 4
    \|w\|_2^2\|xu_t\|_\infty^2 ,
\]
where we used \eqref{comp: sup V x_2}.
Hence
\[
    \|V_t(x_1,x_2)x_1\p_1\Psi_{N,t}\|_2
    \le 2
    \|w\|_2\|xu_t\|_\infty .
\]

For the second term, using $\p_1x_1^2\p_1=\|xu_t\|_2^2\p_1$, we obtain
\[
\begin{aligned}
    \|x_1\p_1V_t(x_1,x_2)\Psi_{N,t}\|_2^2
    &=
    \langle\Psi_{N,t},
        V_t(x_1,x_2)\p_1x_1^2\p_1V_t(x_1,x_2)
        \Psi_{N,t}
    \rangle                                      \\
    &=
    \|xu_t\|_2^2
    \|\p_1V_t(x_1,x_2)\Psi_{N,t}\|_2^2 .
\end{aligned}
\]
Furthermore,
\[
\begin{aligned}
    \|\p_1V_t(x_1,x_2)\Psi_{N,t}\|_2^2
    &\le
    \|u_t\|_\infty^2
    \int |V_t(x_1,x_2)|^2|\Psi_{N,t}(x_1,\dots,x_N)|^2
    \,\dd x_1\cdots\dd x_N                         \\
    &\le
    \|u_t\|_\infty^2
    \sup_{x_1,x_2}|V_t(x_1,x_2)|^2
    \|\Psi_{N,t}\|_2^2                              \\
    &\le 4
    \|w\|_\infty^2\|u_t\|_\infty^2 .
\end{aligned}
\]
Thus
\[
    \|x_1\p_1V_t(x_1,x_2)\Psi_{N,t}\|_2
    \le 2
    \|w\|_\infty\|xu_t\|_2\|u_t\|_\infty .
\]

By Hypothesis \ref{hyp} on the Hartree solution, it holds 
\[
    \|xu_t\|_\infty\lesssim 1,
    \qquad
    \|xu_t\|_2\|u_t\|_\infty\lesssim 1.
\]
Consequently,
\[
    \|V_t(x_1,x_2)x_1\p_1\Psi_{N,t}\|_2
    +
    \|x_1\p_1V_t(x_1,x_2)\Psi_{N,t}\|_2
    \le C,
\]
where $C$ is independent of $N$ and $t$.
\end{proof}

\subsection{Proof of Theorem \ref{thm:intro-many-body-minvel}\eqref{thm:minimal-scattering}}
\label{subsec:scattering}

Theorem \ref{thm:intro-many-body-minvel}\eqref{thm:minimal-scattering} combines the one-body scattering theory for the boson star equation with the many-body minimal velocity estimate proved in Theorem \ref{thm:intro-many-body-minvel}\eqref{thm:minimal-trace-bound}. The basic observation is that the asymptotic velocity operator $\Theta$ commutes with the free Hamiltonian, so that the outgoing velocity distribution is conserved by the free dynamics. The convergence of the reduced density matrix then transfers this information to the many-body evolution, while the minimal velocity estimate identifies the asymptotic velocity with the rescaled position observable.

\begin{proof}[Proof of Theorem \ref{thm:intro-many-body-minvel}\eqref{thm:minimal-scattering}]

Let $\gamma_{N,t}^{(1)}$ denote the reduced one-particle density matrix associated to $\Psi_{N,t}$.
By Theorem~\ref{thm:intro-many-body-mf}, there exists a constant $C>0$ such that
\begin{equation}\label{eq:trace-convergence-proof}
    \Tr \left|
        \gamma_{N,t}^{(1)}
        -
        |u_t\rangle\langle u_t|
    \right|
    \le
    \frac{C}{\sqrt N},
\end{equation}
uniformly in time.

 The scattering assumption implies
\[
\lim_{t\to\infty}\left\|
\e^{\i t\sqrt{-\Delta + 1}}u_t -u_+\right\|_2=0
\]
and since $\Theta$ commutes with $\sqrt{-\Delta + 1}$ we have
\[
\lim_{t\to\infty}
\left\|
    \Id_{\{\Theta^2\ge v^2+\varepsilon\}}
    u_t
\right\|_2
=
\left\|
    \Id_{\{\Theta^2\ge v^2+\varepsilon\}}
    u_+
\right\|_2
=1.
\]
Equivalently,
\[
\lim_{t\to\infty}
\left\|
    \Id_{\{\Theta^2<v^2+\varepsilon\}}
    u_t
\right\|_2
=0.
\]
Hence,
\begin{equation}\label{eq:velocity-tail}
    \limsup_{t\to\infty}
    \Tr\!\left(
        \Id_{\{\Theta^2<v^2+\varepsilon\}}
        \gamma_{N,t}^{(1)}
    \right)
    \le
    \frac{C}{\sqrt N},
\end{equation}
by \eqref{eq:trace-convergence-proof}.

Choose $f,g\in C_c^\infty(\bR)$ such that $0\le f \le 1, 0 \le g \le 1$,
\[
f(s)=1
\quad\text{for }{s\in[0,v^2]},
\qquad
g(s)=1
\quad\text{for }{s\in[ v^2+\varepsilon,1]},
\]
and
\[
\supp f\cap\supp g=\emptyset.
\]
Then
\[
\Id_{\{|x_1|\le vt\}}
\le
f\!\left(\frac{|x_1|^2}{t^2}\right),
\]
and therefore
\begin{align}
\left\|
    \Id_{\{|x_1|\le vt\}}
    \Psi_{N,t}
\right\|_2^2
\le
\Tr\!\left(
    f^2\!\left(\frac{|x|^2}{t^2}\right)
    \gamma_{N,t}^{(1)}
\right)
&=
\Tr\!\left(
    g(\Theta^2)
    f^2\!\left(\frac{|x|^2}{t^2}\right)
    g(\Theta^2)
    \gamma_{N,t}^{(1)}
\right)
\nonumber\\
&\quad+
\Tr\!\left(
    (\Id-g(\Theta^2))
    f^2\!\left(\frac{|x|^2}{t^2}\right)
    g(\Theta^2)
    \gamma_{N,t}^{(1)}
\right)
\nonumber\\
&\quad+
\Tr\!\left(
    f^2\!\left(\frac{|x|^2}{t^2}\right)
    (\Id-g(\Theta^2))
    \gamma_{N,t}^{(1)}
\right).
\label{eq:decomposition-proof}
\end{align}
The first term can be bounded by Theorem \ref{thm:intro-many-body-minvel}\eqref{thm:minimal-trace-bound}.
For the remaining terms, note that
\[
\Id-g(\Theta^2)
\le
\Id_{\{\Theta^2<v^2+\varepsilon\}}.
\]
Since
\(
0\le f^2\le1
\)
and
\(
0\le g\le1,
\)
both terms are bounded by
\[
\Tr\!\left(
    \Id_{\{\Theta^2<v^2+\varepsilon\}}
    \gamma_{N,t}^{(1)}
\right).
\]
Hence we have 
\begin{align*}
\left\|
    \Id_{\{|x_1|\le vt\}}
    \Psi_{N,t}
\right\|_2^2 \le  & C\left(\frac{1}{t^2}
        +\frac{1}{t} 
        +
        \frac{ t^2}{N}
    \right) +2\Tr\!\left(
    \Id_{\{\Theta^2<v^2+\varepsilon\}}
    \gamma_{N,t}^{(1)}
\right).
\end{align*}
Applying \eqref{eq:velocity-tail}, we have that for all $\delta>0$ there exists $t_1(\delta)>0$ such that for all $t>t_1(\delta)$ it holds 
\begin{equation*}
    2\Tr\!\left(
    \Id_{\{\Theta^2<v^2+\varepsilon\}}
    \gamma_{N,t}^{(1)}
\right) \le 2  \limsup_{t\to\infty}
    \Tr\!\left(
        \Id_{\{\Theta^2<v^2+\varepsilon\}}
        \gamma_{N,t}^{(1)}
    \right) + \delta
    \le
    \frac{C}{\sqrt N} + \delta,
\end{equation*}
which yields 
\begin{align*}
    \left\|
    \Id_{\{|x_1|\le vt\}}
    \Psi_{N,t}
\right\|_2^2 \le  & C\left(\frac{1}{t^2}
        +\frac{1}{t} 
        +
        \frac{ t^2}{N} 
    \right) + \frac{C t}{\sqrt{N}} + \delta
\end{align*}
for all $t>t_1(\delta)$.
Since we assume $t=o(N^{1/2})$ as $t, N\to \infty$, there exist $t_2(\delta), N(\delta)>0$ such that for all  $t>t_2(\delta)$ and $N>N(\delta)$ all the remaining terms in the right-hand side are smaller than $\delta$. We conclude that for all $\delta>0$ there exist $T(\delta), N(\delta)>0$ such that 
\begin{equation*}
      \left\|
    \Id_{\{|x_1|\le vt\}}
    \Psi_{N,t}
\right\|_2^2 \le 5\, \delta \qquad \forall\, t>T(\delta), \ N>N(\delta).
\end{equation*}
Consequently,
\[
\lim_{
\substack{
t\to\infty,\\
N\to\infty,\\
t^2/N\to0
}}
\left\|
    \Id_{\{|x_1|\le vt\}}
    \Psi_{N,t}
\right\|_2
=0.
\]
Since
\[
\left\|
    \Id_{\{|x_1|\ge vt\}}
    \Psi_{N,t}
\right\|_2^2
=
1-
\left\|
    \Id_{\{|x_1|\le vt\}}
    \Psi_{N,t}
\right\|_2^2,
\]
the desired conclusion follows.
\end{proof}

\vspace{3em}

\section{Propagation Estimates for the One-body Boson Star Equation}\label{sec:BSE-prop-est}

In this section, we transfer the many-body propagation estimates to the effective boson star equation. The basic idea is simple: we apply the many-body estimates to factorized initial data $\Psi_{N,0}=u_0^{\otimes N}$ and then use the convergence of the one-particle reduced density matrix
\[
    \gamma^{(1)}_{N,t}
    \longrightarrow
    |u_t\rangle\langle u_t|,
\]
which holds in trace norm.

\subsection{Transfer from Many-Body Estimates to the one-body boson star equation}\label{subsec:Nbd21bd}

Let $u_t$ be the solution to the boson star equation \eqref{eq:BSE}, and let $\Psi_{N,t}$ be the solution to the~$N$-body equation \eqref{eq:NbdQSE} with factorized initial datum $\Psi_{N,0}=u_0^{\otimes N}$. Assuming Hypothesis \ref{hyp}, we can apply Theorem \ref{thm:intro-many-body-mf} to get
the mean-field convergence estimate
\begin{equation}\label{eq:transfer-mf-convergence}
    \Tr\left|
        \gamma^{(1)}_{N,t}
        -
        |u_t\rangle\langle u_t|
    \right|
    \leq
    C N^{-1/2}
\end{equation}
uniformly in time.

The following elementary observation will be used repeatedly. If $B$ is a bounded one-particle operator, then
\begin{equation}
    \left|
    \Tr\left[
        B\left(
            \gamma^{(1)}_{N,t}
            -
            |u_t\rangle\langle u_t|
        \right)
    \right]
    \right|
    \leq
    \|B\|_{\op}
    \Tr\left|
        \gamma^{(1)}_{N,t}
        -
        |u_t\rangle\langle u_t|
    \right|
    \leq
    C\|B\|_{\op}N^{-1/2}.
    \label{eq:transfer-bounded-observable}
\end{equation}
In particular, if a many-body estimate gives a bound for $\Tr (B\gamma^{(1)}_{N,t}) $,
then the same bound holds for 
\[
    \Tr ( B|u_t\rangle\langle u_t|  )
    =
    \langle u_t,Bu_t\rangle
\]
after taking $N\to\infty$, provided that $B$ is bounded.

\subsection{Proof of Theorem \ref{thm:intro-BSE-maxvel} for the maximal velocity}\label{subsec:1bd-max-vel}

We first transfer the many-body maximal velocity estimate. 
\begin{proof}[Proof of Theorem \ref{thm:intro-BSE-maxvel}]
Let $X,Y\subset\bR^3$ be convex sets with $\dist(X,Y)>0$. Assume that the initial datum $u_0$ is normalized and satisfies
\[
    \supp u_0\subset X.
\]
Then the factorized initial datum satisfies
\[
    \supp\Psi_{N,0}\subset X^N.
\]
Applying the many-body maximal velocity estimate, see Theorem \ref{thm:maxvel}, gives
\begin{equation}\label{eq:many-body-maxvel-transfer}
    \Tr\big(\Id_Y\gamma^{(1)}_{N,t}\big)
    \leq
    \exp\big(2(t-\dist(X,Y))\big).
\end{equation}
On the other hand,
\begin{equation*}
    \int_Y |u_t(x)|^2\,\d x
    =
    \Tr\big(\Id_Y |u_t\rangle\langle u_t|\big)
    \leq
    \Tr\big(\Id_Y \gamma^{(1)}_{N,t}\big)
    +
    \Tr\left|
        \gamma^{(1)}_{N,t}
        -
        |u_t\rangle\langle u_t|
    \right|.
\end{equation*}
Using \eqref{eq:transfer-mf-convergence} and \eqref{eq:many-body-maxvel-transfer}, we obtain
\[
    \int_Y |u_t(x)|^2\,\d x
    \leq
    \exp\big(2(t-\dist(X,Y))\big)
    +
    C N^{-1/2}.
\]
Letting $N\to\infty$ yields the maximal velocity estimate for the boson star equation:
\begin{equation}\label{eq:BSE-maxvel}
    \int_Y |u_t(x)|^2\,\d x
    \leq
    \exp\big(2(t-\dist(X,Y))\big).
\end{equation}
\end{proof}

\subsection{Proof of Theorem \ref{thm:intro-BSE-minvel} for the minimal velocity}\label{subsec:1bd-min-vel}

We now transfer the many-body minimal velocity estimate. 
\begin{proof}[Proof of Theorem \ref{thm:intro-BSE-minvel}]
Let $f,g\in C_c^\infty(\bR)$ satisfy
\[
    \supp f\cap\supp g=\emptyset.
\]
The many-body minimal velocity estimate, see Theorem \ref{thm:intro-many-body-minvel}\eqref{thm:minimal-trace-bound}, implies that for every time 
\[
    t\to\infty,
    \qquad
    t=o(N^{1/2}),
\]
one has
\begin{equation}\label{eq:many-body-minvel-transfer}
    \Tr\left[
        g(\Theta^2)
        f^2\left(\frac{|x|^2}{t^2}\right)
        g(\Theta^2)
        \gamma^{(1)}_{N,t}
    \right]
    \longrightarrow 0
\end{equation}
as $N\to\infty$.

We define the bounded nonnegative one-particle observable
\[
    B_t
    :=
    g(\Theta^2)
    f^2\left(\frac{|x|^2}{t^2}\right)
    g(\Theta^2).
\]
Since $f,g\in C_c^\infty(\bR)$, both factors are bounded operators on $L^2(\bR^3)$, and therefore
\[
    \|B_t\|_{\op}
    \leq
    \|f\|_\infty^2\|g\|_\infty^2.
\]
By \eqref{eq:transfer-bounded-observable},
\begin{align}
    \left|
    \Tr\left[
        B_t\left(
            \gamma^{(1)}_{N,t}
            -
            |u_t\rangle\langle u_t|
        \right)
    \right]
    \right|
    &\leq
    C\|f\|_\infty^2\|g\|_\infty^2 N^{-1/2}.
    \label{eq:minvel-transfer-error}
\end{align}

Now choose sequences $N$ and $t$ such that
\[
    N,t\to\infty,
    \qquad
    t=o(N^{1/2}).
\]
Applying the many-body minimal velocity estimate \eqref{eq:many-body-minvel-transfer} with \eqref{eq:minvel-transfer-error}, we obtain
\[
    \Tr\left[
        g(\Theta^2)
        f^2\left(\frac{|x|^2}{t^2}\right)
        g(\Theta^2)
        |u_t\rangle\langle u_t|
    \right]
    \longrightarrow 0
\]
as $N,\,t\to\infty$, i.e., 
\begin{equation}\label{eq:BSE-minvel}
    \lim_{t \to \infty}\left\|
        f\left(\frac{|x|^2}{t^2}\right)
        g(\Theta^2)u_t
    \right\|_2^2
    =0.
\end{equation}
\end{proof}

\medskip

\subsection{Proof of Corollary~\ref{cor:many-body-minvel-improved} for the updated many-body minimal velocity}\label{subsec:Nbd-min-vel}

We now provide the proof of improved the many-body minimal velocity estimate. 

\begin{proof}[Proof of Corollary~\ref{cor:many-body-minvel-improved}]
We first fix $t>0$. Letting $N\to\infty$ in Theorem~\ref{thm:intro-many-body-minvel}\eqref{thm:minimal-trace-bound} and using \eqref{eq:transfer-bounded-observable}, we obtain
\begin{equation}\label{eq:one-body-minvel-quantitative}
    \left\|f\left(\frac{|x|^2}{t^2}\right)g(\Theta^2)u_t\right\|_2^2\le C\left(\frac{1}{t^2}+\frac{1}{t}\right).
\end{equation}
The constant $C>0$ is independent of $t$. We insert the decomposition $\p_1+\q_1=\Id$. Since $\|\p_1\Psi_{N,t}\|_2\le1$ and $\|\q_1\Psi_{N,t}\|_2^2=\alpha_N(t)$, by Cauchy--Schwarz we obtain
\begin{align*}
    \left\|f\left(\frac{|x_1|^2}{t^2}\right)g(\Theta_1^2)\Psi_{N,t}\right\|_2^2
    &\le 2\left\|f\left(\frac{|x_1|^2}{t^2}\right)g(\Theta_1^2)\p_1\Psi_{N,t}\right\|_2^2+2\left\|f\left(\frac{|x_1|^2}{t^2}\right)g(\Theta_1^2)\q_1\Psi_{N,t}\right\|_2^2\\
    &\le 2\left\|f\left(\frac{|x|^2}{t^2}\right)g(\Theta^2)u_t\right\|_2^2+2\|f\|_\infty^2\|g\|_\infty^2\alpha_N(t)\\
    &\le C\left(\frac{1}{t^2}+\frac{1}{t}+\frac{1}{N}\right),
\end{align*}
where the last inequality follows from \eqref{eq:one-body-minvel-quantitative} and \eqref{eq:alphaNest}. By the definition of $\gamma^{(1)}_{N,t}$, this proves \eqref{eq:many-body-minvel-improved}.

We next prove \eqref{eq:many-body-scattering-improved}. We choose $f,g\in C_c^\infty(\bR)$ with $0\le f,g\le1$, $f=1$ on $[0,v^2]$, $g=1$ on $[v^2+\varepsilon,1]$ and disjoint supports. Since $\Theta^2$ commutes with $\sqrt{-\Delta+1}$, by \eqref{eq:outgoing-assumption} we have
\[
    (\Id-g(\Theta^2))\e^{-\i t\sqrt{-\Delta+1}}u_+=0.
\]
Using $\p_1+\q_1=\Id$ and \eqref{eq:alphaNest} as above, we obtain
\begin{align*}
    \left\|(\Id-g(\Theta_1^2))\Psi_{N,t}\right\|_2^2
    &\le 2\left\|(\Id-g(\Theta^2))u_t\right\|_2^2+2\alpha_N(t)\\
    &\le 2\left\|u_t-\e^{-\i t\sqrt{-\Delta+1}}u_+\right\|_2^2+\frac{C}{N}.
\end{align*}
Since $f=1$ on $[0,v^2]$, by Cauchy--Schwarz and \eqref{eq:many-body-minvel-improved}, we obtain
\begin{align*}
    \left\|\Id_{\{|x_1|\le vt\}}\Psi_{N,t}\right\|_2^2
    &\le 2\left\|f\left(\frac{|x_1|^2}{t^2}\right)g(\Theta_1^2)\Psi_{N,t}\right\|_2^2+2\left\|(\Id-g(\Theta_1^2))\Psi_{N,t}\right\|_2^2\\
    &\le C\left(\frac{1}{t^2}+\frac{1}{t}+\frac{1}{N}\right)+4\left\|u_t-\e^{-\i t\sqrt{-\Delta+1}}u_+\right\|_2^2.
\end{align*}
By \eqref{eq:hyp-scattering}, the right-hand side tends to zero as $\min\{t,N\}\to\infty$. This proves \eqref{eq:many-body-scattering-improved}.
\end{proof}

\appendix

\section{Alternative Proof of One-Body Minimal Velocity}
\label{sec:one-body-minvel-direct}

In this appendix, we give a proof of a one-body minimal velocity estimate. The result we state is a particular case of \cite[Proposition 8.7]{BFG2025}, but the strategy of proof is different. Our proof follows the lines of the one of the moment bound in Section \ref{subsec:MomentBound}.
One could prove a more general result for $w\in L^1(\mathbb{R}^3)+L^q(\mathbb{R}^3)$, like in \cite[Proposition 8.7]{BFG2025}, but we refrain from doing so for the sake of clarity.
\begin{proposition}Assume Hypothesis~\ref{hyp} holds. Let
\begin{equation}\label{eq:moment-bound-one-body}
    A(t):=x-t\Theta
\quad
\text{and}
\quad
    \mathcal{M}(t)
    :=
    \langle u_t, |A(t)|^2 u_t\rangle = t^2 \left\|
        \left(\frac{x}{t}-\Theta\right)u_t
    \right\|_2^2.
\end{equation}
Then $\mathcal{M}(t)$ is bounded on $[0,\infty)$.
\end{proposition}
The operator $A(t)$ and the corresponding moment $\mathcal M(t)$ are the one-body analogue of~\eqref{def:many-body-pc-observable} and~
\eqref{def:many-body-pc-moment}.
\begin{proof}We use the notations
\[
    H(t)=\sqrt{-\Delta + 1}+V_{\mathrm{eff}}(t,x),
    \qquad
    V_{\mathrm{eff}}(t,x):=(w*|u_t|^2)(x)
\]
and recall that
\[
    \Theta=\frac{-\i\nabla}{\sqrt{-\Delta + 1}}.
\]
We first compute the Heisenberg derivative of $A(t)$. Since
\[
    \i[\sqrt{-\Delta + 1},x]=\Theta,
    \qquad
    [\sqrt{-\Delta + 1},\Theta]=0,
\]
we have
\[
    \partial_t A(t)+\i[\sqrt{-\Delta + 1},A(t)]=0.
\]
Moreover, $V_{\mathrm{eff}}(t,x)$ commutes with $x$. Hence
\[
    \partial_t A(t)+\i[H(t),A(t)]
    =
    -\i t [V_{\mathrm{eff}}(t),\Theta].
\]
Therefore we get
\[
\begin{aligned}
    \frac{\dd}{\dd t}\mathcal M(t)
    &=
    \left\langle
        \left(
            \partial_t A(t)+\i[H(t),A(t)]
        \right)u_t,
        A(t)u_t
    \right\rangle
    \\
    &\quad+
    \left\langle
        A(t)u_t,
        \left(
            \partial_t A(t)+\i[H(t),A(t)]
        \right)u_t
    \right\rangle
    \\
    &=
    \i t
    \left\langle
        [V_{\mathrm{eff}}(t),\Theta]u_t,
        A(t)u_t
    \right\rangle
    -
    \i t
    \left\langle
        A(t)u_t,
        [V_{\mathrm{eff}}(t),\Theta]u_t
    \right\rangle 
\end{aligned}
\]
and 
\[
    \left|\frac{\dd}{\dd t}\mathcal M(t)\right|
    \leq
    2t \, \|[V_{\mathrm{eff}}(t),\Theta]u_t\|_2 \, \mathcal M(t)^{1/2}.
\]
Since $\|\Theta\|_{\mathrm{op}}$ is bounded by 1  and by using Young's inequality, we have,
\[
    \|[V_{\mathrm{eff}}(t),\Theta]u_t\|_2
    \leq
    \|w\|_{1} \, 
    \|u_t\|_2 \,
    \|u_t\|_\infty^2.
\]
Hence,
\[
    \left|\frac{\dd}{\dd t}\mathcal M(t)\right|
    \leq
    2 t \, \|w\|_{1} \, \|u_t\|_\infty^2 \, \mathcal M(t)^{1/2}.
\]
Using the decay estimate from \eqref{eq:decay},
\[
    \left|\frac{\dd}{\dd t}\mathcal M(t)\right|
    \leq
    C t\langle t \rangle^{-3} \mathcal M(t)^{1/2}
    .
\]
Setting $Y(t):=\mathcal M(t)^{1/2}$,
we obtain
\[
    \frac{\dd}{\dd t}Y(t)
    \leq
    C \langle t \rangle^{-2}.
\]
It thus follows that $Y(t)\le Y(0)+C$ and that $\mathcal{M}(t)$ is bounded.
\end{proof}

\section{Time decay of the solution to the one-body equation}
\label{app: t decay}

In this Appendix we prove that the first inequality in assumption \eqref{eq:weighted1bodydecay} from Theorem \ref{thm:intro-many-body-minvel}\eqref{thm:minimal-nu-bound} is satisfied for a suitable class of potential $w$ and initial data $u_0$. 

We recall that $\|f\|_{H^{s,1}}= \|\langle \nabla\rangle^s f\|_1$, with $\langle \nabla \rangle = \sqrt{-\Delta + 1}$.

\begin{lemma}\label{lemma:xu_t L infty}
    Let $w \in L^1(\bR^3)$ and $u_t$ the solution to \eqref{eq:BSE} with $u_0 \in H^{s,1}(\bR^3)$ and $\langle x \rangle u_0 \in H^s(\bR^3)$ with $s\ge 5/2$. If $\|w\|_1 \|u_0\|_{H^s\cap H^{s,1}}$ is sufficiently small, then 
    \begin{equation*}
    \|\langle x\rangle u_t\|_\infty \lesssim 1 , 
    \end{equation*}
    uniformly in $t\ge0$.
\end{lemma}
\begin{proof}
    We write the solution of equation \eqref{eq:BSE} via Duhamel formula 
    \begin{equation*}
        u_t = \e^{-\i t\sqrt{-\Delta + 1}}u_0 - \i \int_0^t  \e^{-\i(t-\tau)\sqrt{-\Delta + 1}}w*|u_\tau|^2 u_\tau \, \d \tau
    \end{equation*}
    and we introduce the notation $\|f\|_{H^s_1}= \|\langle x \rangle f\|_{H^s}$.
    Applying \cite[Lemma B.7]{BFG2025} we have 
    \begin{align*}
     \|\langle x\rangle u_t\|_\infty \lesssim & \langle t \rangle^{-1/2} \|u_0\|_{H^{s,1}} + \|u_0\|_{H^s_1}  \\
     & + \int_0^t (\langle t -\tau\rangle^{-1/2} \|w*|u_\tau|^2 u_\tau\|_{H^{s,1}} +  \|w*|u_\tau|^2 u_\tau\|_{H^{s}_1} )\, \d \tau. 
    \end{align*}
     By H\"older and Young inequalities, combined with a fractional Leibniz rule \cite[Proposition 3.1]{BFG2025}, and since $\|u_\tau\|_2 = \|u_0\|_2=1$ we obtain 
    \begin{align}
        \|w*|u_\tau|^2 u_\tau\|_{H^{s,1}} \lesssim & \|w*|u_\tau|^2\|_2 \|u_\tau\|_{H^s} + \|w*\langle \nabla\rangle^s(|u_\tau|^2)\|_2 \nonumber \\
        \lesssim &\|w\|_1 \|u_\tau\|_\infty\|u_\tau\|_{H^s} \lesssim \|w\|_1 \langle \tau\rangle^{-3/2}, \label{comp: H^s,1}
    \end{align}
    where we applied \cite[Theorem 2.10]{BFG2025} in the last inequality. Thanks to the fact that $\|f\|_{H^s_1}\lesssim \|\langle x \rangle \langle \nabla \rangle^s f\|_2\lesssim \|f\|_{H^s_1}$ we obtain in a similar way the bound 
    \begin{align}
       \| w*|u_\tau|^2 u_\tau\|_{H^{s}_1} \lesssim & \|w*|u_\tau|^2 \|_\infty \|u_\tau\|_{H^{s}_1} + \|w*\langle \nabla \rangle ^s(|u_\tau|^2) \|_2 \|\langle x \rangle u_\tau\|_\infty\nonumber\\
       \lesssim & \|w\|_1 (\|u_\tau \|_\infty^2 \|u_\tau\|_{H^{s}_1} + \|u_\tau \|_\infty \|u_\tau\|_{H^{s}}\|\langle x \rangle u_\tau\|_\infty)\nonumber\\
       \lesssim & \|w\|_1 (\langle \tau\rangle^{-3} \|u_\tau\|_{H^s_1} + \langle \tau\rangle^{-3/2} \|\langle x \rangle u_\tau\|_\infty), \label{comp: H^s_1}
    \end{align}
    where $\|u_\tau\|_{H^s}$ is uniformly bounded in $\tau$ thanks to \cite[Theorem 2.10]{BFG2025}.

    From \cite[Lemma 8.2]{BFG2025} we have $\|u_\tau\|_{H^s_1}\lesssim \|u_0\|_{H^s_1} + \langle \tau \rangle  \|u_0\|_{H^{s,1}}$, combining this with \eqref{comp: H^s,1} and \eqref{comp: H^s_1} we obtain 
    \begin{align*}
    \|\langle x\rangle u_t\|_\infty \lesssim & \, 1 + \|w\|_1\int_0^t ( \langle t -\tau\rangle^{-1/2} \langle \tau \rangle^{-3/2} + \langle \tau \rangle^{-2} + \langle \tau\rangle^{-3/2} \|\langle x \rangle u_\tau\|_\infty)\, \d\tau.
    \end{align*}
    Since $\sup_{t>0} \int_0^t\langle t -\tau\rangle^{-1/2} \langle \tau \rangle^{-3/2} <\infty $, the statement follows from to Gr\"onwall inequality. 
\end{proof}

\bibliographystyle{abbrv}
\bibliography{refs}

\end{document}